\documentclass[reqno]{amsart}

\usepackage[utf8]{inputenc}
\usepackage[T1]{fontenc}
\usepackage{amsmath,amsfonts,amssymb,amsthm}
\usepackage{color}
\usepackage{textcmds}
\usepackage{graphicx}

\newtheorem{theorem}{Theorem}

\newtheorem{lemma}[theorem]{Lemma}

\newtheorem{remark}[theorem]{Remark}
\newtheorem{example}[theorem]{Example}
\newtheorem{scheme}{Scheme}
\newtheorem{testcase}{Benchmark test}

\begin{document}

\title[Commutator-Based Operator Splitting for Linear Port-Hamiltonian Systems]{Commutator-Based Operator Splitting\\ for Linear Port-Hamiltonian Systems}

\author[M. Mönch et al.]{Marius Mönch$^{1,\star}$ \and Nicole Marheineke$^1$}

\date{\today\\
$^1$ Trier University, Arbeitsgruppe Modellierung und Numerik, Universit\"atsring 15, D-54296 Trier, Germany\\
$^\star$ corresponding author, email: moench@uni-trier.de, orcid: 0009-0000-2582-2199
}

\begin{abstract}
In this paper, we develop high-order splitting methods for linear port-Hamiltonian systems, focusing on preserving their intrinsic structure, particularly the dissipation inequality. Port-Hamiltonian systems are characterized by their ability to describe energy-conserving and dissipative processes, which is essential for the accurate simulation of physical systems. For autonomous systems, we introduce an energy-associated decomposition that exploits the system’s energy properties. We present splitting schemes up to order six. In the non-autonomous case, a port-based splitting is employed. This special technique makes it possible to set up methods of arbitrary even order. Both splitting approaches are based on the properties of the commutator and ensure that the numerical schemes not only preserve the structure of the system but also faithfully fulfill the dissipation inequality. The proposed approaches are validated through theoretical analysis and numerical experiments.
\end{abstract}


\maketitle

\section{Introduction}
Finite-dimensional port-Hamiltonian (pH) systems typically arise from port-based network modeling with coupled subsystems comprising different multi-physical domains and/or time scales or from mimetic space
discretizations of partial differential systems. The systems are passive, and physical properties, such as energy conservation or dissipation, are encoded in the algebraic structures, cf., e.g., \cite{van2014port,van2004port,mehrmann2023control}.
The focus of this work is on linear pH-systems of the form
\begin{align*}
    \dot{x}(t) &= (J - R)\nabla H(x(t)) + Bu(t), \qquad x(t_0) = x_0,\\
    y(t) &= B^T \nabla H(x(t))
\end{align*}
with structure matrix $J=-J^T$,  dissipation matrix $R=R^T\geq0$ and port matrix $B$.
The Hamiltonian $H$ is assumed to be quadratic, i.e., $H(x)=\tfrac{1}{2}x^TQx$ with $Q=Q^T>0$ (w.l.o.g.\ we consider $Q=I$ identity). Supposing sufficient regularity of the input $u$, the pH-system has a unique solution $x\in \mathcal{C}^1([0,T],\mathbb{R}^n)$ which satisfies the dissipation inequality,
\begin{align*}
    \tfrac{\mathrm{d}}{\mathrm{d}t}H(x(t))\leq y(t) ^Tu(t).
\end{align*}
Energy conservation particularly holds in the case of $u\equiv 0$ and $R=0$, i.e., $H(x(t))=H(x_0)$.

This work aims at the development and investigation of higher order splitting methods that account for the port-Hamiltonian dissipation inequality. Operator splitting takes advantageous of the decomposition of the underlying problem into subproblems of profoundly different behavior regarding, e.g., dynamics, stiffness or computational costs. In a splitting scheme, the subproblems are then solved sequentially and repeatedly with different step sizes,
\begin{align*}
\dot x&=f(x)=f^{[1]}(x)+f^{[2]}(x),\\
\Phi_h &= \varphi_{b_{m}h}^{[2]} \circ \varphi_{a_{m}h}^{[1]} \circ \varphi_{b_{m-1}h}^{[2]} \circ \cdots \circ \varphi_{a_{2}h}^{[1]} \circ \varphi_{b_{1}h}^{[2]} \circ \varphi_{a_{1}h}^{[1]},\\
& \quad \,
\varphi_{h}^{[i]}=\exp(h\mathcal{L}^{[i]})\mathrm{id}, \quad  \mathcal{L}^{[i]}=\langle f^{[i]}(\cdot), \nabla_x\rangle
\end{align*}
satisfying $\sum_{j=1}^ma_j=\sum_{j=1}^mb_j=1$, cf., e.g., \cite{hairer2006structure,mclachlan2002splitting}. For general $f^{[1]}$ and $f^{[2]}$ classical splitting schemes of order $p\geq 3$ involve negative step sizes (at least one
coefficient $a_j$ and at least one coefficient $b_j$ are strictly negative), \cite{blanes2005necessity,celledoni2018passivity}. For time-irreversible systems, such as pH-systems with $R\neq 0$, the positivity of the step sizes is essential. There exist splitting schemes with complex-valued coefficients with strictly positive real parts that allow for higher order $p\geq3$, cf.\ \cite{blanes2010splittingcomplex,castella2009splitting}.
Alternatively, negative step sizes can be avoided by help of commutator-based methods, \cite{kieri2015stiff}. For splitting methods of fourth order with positive step sizes, see, e.g., the force gradient approaches in \cite{omelyan2002construction}, the force gradient symplectic integrators in \cite{chin2001fourth,chin2000higher} or the force gradient nested multirate methods in \cite{shcherbakov2015force,shcherbakov2017adapted}. In these papers the focus is on Hamiltonian systems. The Hamiltonian is split, resulting in a commutator with a special gradient structure which allows the preservation of energy or the preservation of a shadow Hamiltonian.
In the splitting for pH-systems we face  an additional challenge with the dissipation inequality. The preservation of dissipation requires either force gradient terms with specific properties or vanishing commutators due to appropriately chosen decompositions. The choice of the decomposition is problem-dependent. So far, energy-associated and subsystem (dimension-reducing) decompositions have been investigated for pH-systems of ordinary and differential algebraic equations in a second order Strang splitting \cite{bartel2024,frommer2023operator}.

In this work we develop fourth and sixth-order energy- and dissipation-preserving symmetric methods for linear autonomous pH-systems based on an energy-associated decomposition. The energy-associated decomposition (also referred to as $J$-$R$-decomposition in \cite{bartel2024}) considers an energy-conserving subproblem and a dissipative subproblem and yields commutators with special {(skew-)}symmetry properties. For sixth order, we particularly exploit linear combinations and order-preserving approximations of the commutators, using the Baker-Campbell-Hausdorff formula \cite{hairer2006structure} and the structure of positive partitions of exponential operators \cite{chin2005structure}. The intuitive extension of the energy-associated decomposition to an non-autonomous (port) setting
---by incorporating the time-dependent input into the dissipative subproblem--- allows the preservation
of dissipation in case of Strang splitting ($p = 2$), cf. \cite{frommer2023operator}, but not in higher order. The occurring
commutators are differential operators that lose the port-Hamiltonian structural
properties. We address this problem by introducing a novel port-based decomposition. The resulting splitting strategy has similarities with the "frozen time" technique discussed in \cite{blanes2006splitting,blanes2010splitting,blanes2012time}. We obtain splitting schemes of fourth, sixth and higher even order. The dissipation inequality is accounted for in a quadrature-type version. In addition, we extend
the results from the energy-associated splitting of the autonomous case to the non-autonomous case
by adapting the formula for the variation of constants, and derive a fourth-order method. A similar procedure has been applied for
methods of first and second order in \cite{faou2015analysis,ostermann2012error}.

The paper is structured as follows: Section~\ref{sec:splitting} gives a brief introduction into the concept of splitting and its convergence theory. Fourth and sixth-order structure-preserving schemes are derived for autonomous pH-systems based on an energy-associated decomposition in Section~\ref{sec:aut}. The extension to non-autonomous pH-systems is discussed in Section~\ref{sec:nonaut}. We develop port-based splitting methods of fourth and higher even order as well as a fourth-order scheme that combines an energy-associated splitting with an quadrature rule. The numerical performance of the schemes regarding convergence order and dissipation inequality is discussed in Section~\ref{sec:numerics}. 

\section{Splitting Concept}\label{sec:splitting}
The splitting concept considers a decomposition of an autonomous finite-dimensional differential system into subproblems, cf., e.g., \cite{hairer2006structure,mclachlan1995numerical,mclachlan2002splitting}. Assuming sufficient regularity of the right-hand side, the exact flow $\varphi$ of the original system as well as the flows $\varphi^{[i]}$ of the subproblems $f^{[i]}$ can be expressed  in terms of Lie derivatives $\mathcal{L}^{[i]}$, $i=1,2$, i.e.,
\begin{align*}
\dot{x}&=f(x)=f^{[1]}(x)+f^{[2]}(x), \qquad x(0)=x_0,\\
\varphi_h(x_0)&= \exp(h(\mathcal{L}^{[1]}+\mathcal{L}^{[2]}))\, \mathrm{id}(x_0), \qquad \varphi_h^{[i]}= \exp(h \mathcal{L}^{[i]})\,\mathrm{id}, \qquad \mathcal{L}^{[i]}= \langle f^{[i]}(\cdot), \nabla_x\rangle
\end{align*}
with  $\mathrm{id}$ identity mapping. A general splitting scheme 
\begin{align*}
\Phi_h = \varphi_{b_{m}h}^{[2]} \circ \varphi_{a_{m}h}^{[1]} \circ \varphi_{b_{m-1}h}^{[2]} \circ \cdots \circ \varphi_{a_{2}h}^{[1]} \circ \varphi_{b_{1}h}^{[2]} \circ \varphi_{a_{1}h}^{[1]}
\end{align*}
of $2m$ stages with step sizes $a_j$, $b_j$, $j=1,...,m$ consists of expressions  $\varphi_{b_{j}h}^{[2]} \circ \varphi_{a_{j}h}^{[1]} $  which can be written as an exponential in the form of
\begin{align*}
       \varphi_{b_jh}^{[2]} \circ \varphi_{a_jh}^{[1]} &=  \exp(a_jh \mathcal{L}^{[1]})\exp(b_jh \mathcal{L}^{[2]} )\,\mathrm{id}\\
       &= \exp(\,h (a_j\mathcal{L}^{[1]}+b_j \mathcal{L}^{[2]}) 
      + h^2 \tfrac{1}{2} a_jb_j [\mathcal{L}^{[1]}, \mathcal{L}^{[2]}]\\
       &\qquad \quad  +h^3 \tfrac{1}{12} (a_j^2 b_j [\mathcal{L}^{[1]},\mathcal{L}^{[1]}, \mathcal{L}^{[2]}]]-a_j b_j^2 [\mathcal{L}^{[2]},\mathcal{L}^{[1]}, \mathcal{L}^{[2]}]])+\mathcal{O}(h^4)\,)\,\mathrm{id}
    \end{align*}
according to the Lemma by Gröbner and the Baker-Campbell-Hausdorff formula, cf.\ \cite{hairer2006structure}.
The Lie bracket prescribes a linear differential operator, i.e., $[\mathcal{L}^{[1]}, \mathcal{L}^{[2]}]=\mathcal{L}^{[1]} \mathcal{L}^{[2]}-\mathcal{L}^{[2]} \mathcal{L}^{[1]}$. Hence, a splitting method can be considered as an exponential of differential operators (commutators), it is consistent if the chosen step sizes satisfy $\sum_j a_j=\sum_j b_j=1$.

We aim for schemes of high order $p\gg1$, i.e., $\Phi_h= \exp(h(\mathcal{L}^{[1]}+\mathcal{L}^{[2]})+\mathcal{O}(h^{p+1}))\,\mathrm{id}$, that account for the port-Hamiltonian dissipation inequality. Our focus is on symmetric schemes, where the step sizes obey either case a) or case b) --- depending on which subproblem is solved first, i.e.,
\begin{align}\label{scheme:sym}
 \Phi_h = \bigg(\prod_{j=1}^m  \exp(a_jh \mathcal{L}^{[1]})\exp(b_jh \mathcal{L}^{[2]} )\bigg)\,\mathrm{id}, \qquad s=\lceil \tfrac{m}{2}\rceil
 \end{align}
  \begin{itemize}
 \item[a)] \hspace*{1.1cm} $a_i=a_{m+1-i}$, \, $i=1,...,s$, \hspace*{1cm} $b_m=0$, $b_i=b_{m-i}$,   \hspace*{0.25cm} \, $i=1,...,s-1$, 
  \item[b)] $a_1=0$, $a_i=a_{m+2-i}$, \, $i=2,...,s$, \hspace*{2.3cm} $b_i=b_{m+1-i}$,  \, $i=1,...,s$. 
 \end{itemize}
 Symmetric schemes posses an even order and can be formally written in the form
\begin{align}\label{scheme:sym_order}
       \Phi_h&=\exp(\,h\, (e_{1_1}\mathcal{L}^{[1]}+e_{1_2} \mathcal{L}^{[2]}) +h^3\, (e_{3_1} \mathcal{C} + e_{3_2} \mathcal{D})\\  \nonumber
       & \quad \qquad +
       h^5\, (\,\, e_{5_1} [\mathcal{L}^{[1]},[\mathcal{L}^{[1]},\mathcal{C}]]+e_{5_2} [\mathcal{L}^{[1]},\mathcal{L}^{[1]},\mathcal{D}]] +
       e_{5_3} [\mathcal{L}^{[1]},[\mathcal{L}^{[2]},\mathcal{D}]]
       \\ 
       & \quad \qquad \qquad +e_{5_4} [\mathcal{L}^{[2]},[\mathcal{L}^{[1]},\mathcal{C}]] +
       e_{5_5}  [\mathcal{L}^{[2]},[\mathcal{L}^{[2]},\,\mathcal{C}]]+e_{5_6} [\mathcal{L}^{[2]},[\mathcal{L}^{[2]},\mathcal{D}]]\,\,) + \mathcal{O}(h^7)\,)\, \mathrm{id} \nonumber
\end{align}
with the commutators 
\begin{align*}
\mathcal{C}=[\mathcal{L}^{[1]},[\mathcal{L}^{[1]},\mathcal{L}^{[2]}]], \qquad \mathcal{D}=[\mathcal{L}^{[2]},[\mathcal{L}^{[1]},\mathcal{L}^{[2]}]].
\end{align*}
The relation $[\mathcal{L}^{[1]}, \mathcal{D}]= -[\mathcal{L}^{[2]}, \mathcal{C}]$ holds because of the Jacobi identity, i.e., $[\mathcal{L}^{[A]},[\mathcal{L}^{[B]},\mathcal{L}^{[C]}]]+[\mathcal{L}^{[B]},[\mathcal{L}^{[C]},\mathcal{L}^{[A]}]]+[\mathcal{L}^{[C]},[\mathcal{L}^{[A]},\mathcal{L}^{[B]}]]=0$. The coefficients $e_{k_\ell}$ depend on the step size tuples $a$ and $b$, in particular $e_{1_1}=\sum_j a_j$ and $e_{1_2}=\sum_j b_j$.  To fulfill the conditions for order $p$, i.e., 
\begin{align}\label{scheme:order}
e_{1_1}=e_{1_2}=1 \, \, \text{ and } \,\, e_{k_\ell}=0  \text{ for } k=3,5,...,p-1 \text{ and all } \ell,
\end{align} 
a symmetric scheme has $m$ degrees of freedom in $(2m-1)$ stages. 

\begin{example}[Strang splitting] The Strang splitting is certainly the best-known second-order method, as it has the smallest number of stages (degrees of freedom, $m=2$), i.e.,
 \begin{align*}
 \text{a) } &a_1=\tfrac{1}{2}, \quad b_1=1,  \qquad  \qquad \qquad \qquad \text{b) } a_2=1, \quad b_1=\tfrac{1}{2}, \\
 &\Rightarrow \, \Phi_h = \varphi_{h/2}^{[\alpha]} \circ \varphi_{h}^{[\beta]} \circ \varphi_{h/2}^{[\alpha]}, \qquad \alpha,\beta \in\{1,2\}, \,\, \alpha\neq \beta,
 \end{align*}
cf.\ \eqref{scheme:sym}. It preserves the dissipation inequality for general non-linear port-Hamiltonian ordinary differential systems, including a general Hamiltonian, state-dependent system matrices and time-dependent input, when energy-associated or subsystem-/port-based decompositions are used, \cite{frommer2023operator}.
\end{example}

Splitting schemes of order $p\geq 3$ for general $f^{[1]}$ and $f^{[2]}$ involve negative step sizes (at least one coefficient $a_j$ and at least one coefficient $b_j$ are strictly negative, cf., e.g., \cite{blanes2005necessity, celledoni2018passivity}), by which the port-Hamiltonian dissipation inequality is violated. As an example, we refer to the fourth-order triple-jump method ($m=4$) given, for $\alpha,\beta \in\{1,2\} $, $\alpha\neq \beta$, by
\begin{align*}
\Phi_h = \varphi_{\tfrac{\gamma_1}{2}h}^{[\alpha]} \circ \varphi_{\gamma_1h}^{[\beta]} \circ \varphi_{\tfrac{\gamma_1+\gamma_2}{2}h}^{[\alpha]} \circ \varphi_{\gamma_2 h}^{[\beta]} \circ \varphi_{\tfrac{\gamma_1+\gamma_2}{2}h}^{[\alpha]} \circ \varphi_{\gamma_1h}^{[\beta]}  \circ \varphi_{\tfrac{\gamma_1}{2}h}^{[\alpha]}, \quad \gamma_1=\tfrac{1}{2-\sqrt[3]{2}}>0, \,\gamma_2=\tfrac{-\sqrt[3]{2}}{2-\sqrt[3]{2}}<0.
\end{align*}

For time-irreversible systems, such as pH-systems, the positivity of the step sizes is essential. There exist splitting schemes  with complex-valued coefficients with strictly positive real parts that allow for higher order $p\geq3$, cf.\ \cite{blanes2010splittingcomplex,castella2009splitting}.
In this work, we demand for real-valued positive step sizes. To set up high-order dissipation-preserving schemes for port-Hamiltonian systems, we explore the commutators. We follow the strategy that we satisfy as many order conditions as possible by the choice of positive step sizes and cancel the remaining terms either by vanishing commutators due to appropriately chosen decompositions or by forced gradient terms with specific properties.

\section{Autonomous PH-Systems with Energy-associated Decomposition}\label{sec:aut}
In this section we consider linear autonomous port-Hamiltonian systems with an energy-associated decomposition into a dissipative and an energy-conserving subproblem, used, e.g., in \cite{bartel2024,frommer2023operator}, i.e., 
\begin{align}\label{model:autonom} 
    &\dot{x} = (J - R)x=f^{[1]}(x) +f^{[2]}(x) , && x(0)=x_0,\\
    &f^{[1]}(x) =\bar{R}x, \,\, \bar{R}=-R\leq0, \qquad   f^{[2]}(x) =Jx, \,\, J=-J^T. && \nonumber
\end{align}
The pH-system has a quadratic Hamiltonian $H(x)=\frac{1}{2}x^Tx$, and the dissipation inequality is satisfied by its unique solution, i.e. $\tfrac{\mathrm{d}}{\mathrm{d}t}H(x(t))\leq 0$.
The exact flow of the original system is in particular given by $\varphi_h(x_0)=\exp(h(\bar R+J))x_0$. Analogously, we have $\varphi_h^{[1]}=\exp(h\bar R)$ and $\varphi_h^{[2]}=\exp(h J)$ for the subproblems. 

For the derivation of splitting schemes, the Baker-Campbell-Hausdorff formula on the composition of exponentials is applicable and the order conditions are the same as in the general (nonlinear) case \eqref{scheme:order}. In the considered linear setting, however, the system matrices $L^{[1]} =\bar{R}$ and $L^{[2]} =J$ that are in general non-commuting, take formally the place of the Lie derivatives in \eqref{scheme:sym_order}. Hence, the commutators are here not differential operators, but matrices. To distinguish  between differential operators and matrices in the notation, we use a calligraphic and a normal font.

\begin{lemma} \label{lem:com} Let $L^{[1]} $ be a symmetric matrix and  $L^{[2]}$ a skew-symmetric matrix. Then a commutator in \eqref{scheme:sym_order} is skew-symmetric if it contains an odd number of $L^{[2]}$, otherwise it is symmetric. In particular, $C=[{L}^{[1]},[{L}^{[1]},{L}^{[2]}]] $ is skew-symmetric and ${D}=[{L}^{[2]},[{L}^{[1]},{L}^{[2]}]]$ is symmetric.
\end{lemma}
\begin{proof}
Let $X$, $Y$ and $Z$ be symmetric or skew-symmetric matrices. Consider the commutator $K=[X,[Y,Z]]$ and let $k$ be the number of involved skew-symmetric matrices in $K$. From
$$K=XYZ-XZY-YZX+ZYX$$
we straightforward obtain
\begin{align*}
K^T&= Z^T Y^T X^T -Y^T Z^T X^T-X^T Z^T Y^T+X^T Y^T Z^T\\
      &= (-1)^k (ZYX-YZX-XZY+XYZ)=(-1)^k K.
\end{align*}
\end{proof}

\begin{remark}
For general port-Hamilto\-nian system matrices $\bar R\leq 0$ and $J=-J^T$, the symmetric commutators occurring in \eqref{scheme:sym_order} are non-definite. 
\end{remark}

\subsection{Fourth-order dissipation-preserving schemes}
In view of the dissipation preservation, we demand for positive step sizes and hence cannot fulfill all fourth order conditions. In particular, the terms $e_{3_1}$ and $e_{3_2}$ cannot be canceled out simultaneously, cf.\ \cite{chin2005structure,suzuki1991general}. To derive a fourth-order scheme, we make use of the skew-symmetry of the commutator $C=[\bar R, [\bar R, J]]$ and treat it as a forced gradient term.

We fulfill the following order conditions by choosing appropriate step size tuples $a$ and $b$,
$$e_{1_1}(a,b)=e_{1_2}(a,b)=1, \qquad e_{3_2}(a,b)=0.$$
A symmetric scheme which provides enough degrees of freedom ($m=3$) has a stage number of at least five. The satisfaction of the order conditions involves the solve of a non-linear system. Then, we evaluate the remaining coefficient $c=e_{3_1}(a,b)$ and incorporate the term $(-ch^3C)$ into the scheme. The incorporation can be performed in different ways, but must account for the symmetric structure of the scheme. Formally, we obtain two classes of symmetric schemes with $m\geq 3$  in accordance to the cases a) and b) in \eqref{scheme:sym}
\begin{align}\label{scheme:sym-force}
 \Phi_h = \prod_{j=1}^m  \exp(a_jh L^{[1]}-c_j^{[1]} h^3C)\,\exp(b_jh {L}^{[2]}-c_j^{[2]}h^3C),  
 \qquad \sum_{j=1}^m c_j^{[1]}+c_j^{[2]}= c, 
 \end{align}
 \begin{itemize}
 \item[a)] \hspace*{1.25cm} $c_i^{[1]}=c_{m+1-i}^{[1]}$, \,$i=1,...,s$ \, as $a_i$, \hspace*{0.55cm} $c_m^{[2]}=0$, $c_i^{[2]}=c_{m-i}^{[2]}$, \hspace*{0.20cm} \,$i=1,...,s-1$ \, as $b_i$, 
  \item[b)] $c_1^{[1]}=0$,  $c_i^{[1]}=c_{m+2-i}^{[1]}$, \,$i=2,...,s$ \, as $a_i$, \hspace*{1.9cm} $c_i^{[2]}=c_{m+1-i}^{[2]}$,  \,$i=1,...,s$ \, as $b_i$, 
 \end{itemize}
for all $a_j, b_j \geq 0$ and $s=\lceil \tfrac{m}{2}\rceil$.

\begin{lemma} A commutator-based scheme of the form \eqref{scheme:sym-force} preserves the dissipation inequality of \eqref{model:autonom}.
\end{lemma}
\begin{proof}
Let $M^{[1]}_j=\bar R - \tfrac{c_j^{[1]}}{a_j}h^2 C$ and $M^{[2]}_j=J-\tfrac{c_j^{[2]}}{b_j}h^2 C$. A step of the splitting scheme is performed via the solves of $2m$ differential systems, i.e., $\Phi_h(x_0)=x_{m}^{[2]}(h_m^{[2]})$ with 
\begin{align*}
\dot x_1^{[1]}&=M^{[1]}_1 x_1^{[1]}, &&  x_1^{[1]}(0)=x_0,\\
\dot x_1^{[2]}&=M^{[2]}_1 x_1^{[2]}, &&   x_1^{[2]}(0)=x_1^{[1]}({h}_1^{[1]}),\\
& \, \vdots  && \hspace*{1.cm}\vdots \\
\dot x_{m}^{[1]}&=M^{[1]}_m x_{m}^{[1]}, &&  x_{m}^{[1]}(0)=x_{m-1}^{[2]}({h}_{m-1}^{[2]}),\\
\dot x_{m}^{[2]}&=M^{[2]}_m x_{m}^{[2]}, &&  x_{m}^{[2]}(0)=x_{m}^{[1]}({h}_m^{[1]}),
\end{align*} 
where $h_j^{[1]}=a_j h\geq 0$ and $h_j^{[2]}=b_j h\geq 0$.
For the quadratic Hamiltonian we obtain
\begin{align*}
H(x_{m}^{[2]}(h_m^{[2]}))-H(x_0)=
\sum_{j=1}^m \bigg(H(x_{j}^{[2]}(h_j^{[2]}))-H(x_{j}^{[2]}(0))\bigg)+
\bigg(H(x_{j}^{[1]}(h_j^{[1]}))-H(x_{j}^{[1]}(0))\bigg)\leq 0,
\end{align*}
since the estimates
\begin{align*}
H(x_{j}^{[1]}(h_j^{[1]}))-H(x_{j}^{[1]}(0))&=\int_0^{a_j h} \tfrac{\mathrm{d}}{\mathrm{d}t}H(x_j^{[1]}(t)) \, \mathrm{d}t=\int_0^{a_j h} (x_j^{[1]})^T \, (\bar R -\tfrac{c_j^{[1]}}{a_j}h^2 C) \, x_j^{[1]}\, \mathrm{d}t \leq 0,\\
H(x_{j}^{[2]}(h_j^{[2]}))-H(x_{j}^{[2]}(0))&=\int_0^{b_j h} \tfrac{\mathrm{d}}{\mathrm{d}t}H(x_j^{[2]}(t)) \, \mathrm{d}t=\int_0^{b_j h} (x_j^{[2]})^T \, (J -\tfrac{c_j^{[2]}}{b_j}h^2 C) \, x_j^{[2]} \, \mathrm{d}t = 0
\end{align*}
hold for all $j$ due to the negative definiteness of $\bar R$ and the skew-symmetry of $J$ and $C$.
\end{proof}

\begin{scheme}[5-stage schemes]\label{scheme:5stage} The two classes of dissipation-preserving 5-stage schemes for \eqref{model:autonom} are given by\\[0.7ex]
\begin{tabular}{l l l l l}
\qquad a) & $a_1=\tfrac{1}{6},$ & $a_2=\tfrac{2}{3},$ & $b_1=\tfrac{1}{2},$ & $c =e_{3_1}=\tfrac{1}{6}b_1(a_2^2-2a_1a_2-2a_1^2)=\tfrac{1}{72}$\\[0.5ex]
\qquad b) & $a_2=\tfrac{1}{2},$ & $b_1=\tfrac{3-\sqrt{3}}{6},$ & $b_2=\tfrac{\sqrt{3}}{3},$ & $c=e_{3_1}=-\tfrac{1}{6}a_2^2(b_2-4b_1)=\tfrac{2-\sqrt{3}}{24}$
\end{tabular}\\[0.7ex]
Class-a schemes start with a solve of subproblem $f^{[1]}$ (hence $b_m=0$), class-b schemes with a solve of subproblem $f^{[2]}$ (hence $a_1=0$). The remaining coefficients obey the symmetry conditions stated in \eqref{scheme:sym}.

The underlying non-linear systems for the order conditions of the two classes are \\[0.5ex]
\begin{tabular}{l l l l }
a) & $e_{1_1}= 2a_1+a_2=1,$ & $e_{1_2}=2b_1=1,$ & $e_{3_2}=\tfrac{1}{6}b_1^2(a_2-4a_1)=0$ \\[0.3ex]
b) & $e_{1_1}=2a_2=1,$ & $e_{1_2}=2b_1+b_2=1,$ & $e_{3_2}= -\tfrac{1}{6}a_2(b_2^2-2b_1b_2-2b_1^2)=0$
\end{tabular}
\end{scheme}

\begin{example}[5-stage commutator-based methods]\label{ex:4order}
\begin{align*}
(i)\phantom{ii} \qquad \Phi_h &=  \exp(\tfrac{h}{6} \bar R) \,\, \exp(\tfrac{h}{2} J) \,\,  \exp(\tfrac{2h}{3} \bar R-\tfrac{h^3}{72}C) \,\, \exp(\tfrac{h}{2} J) \,\, \exp(\tfrac{h}{6} \bar R)\\
(ii)\phantom{i} \qquad \Phi_h &=  \exp(\tfrac{h}{6} \bar R) \,\, \exp(\tfrac{h}{2} (J-\tfrac{h^2}{72}C)) \,\, \exp(\tfrac{2h}{3} \bar R) \,\, \exp(\tfrac{h}{2}(J -\tfrac{h^2}{72}C)) \,\, \exp(\tfrac{h}{6} \bar R)\\
(iii) \qquad \Phi_h &=  \exp(\tfrac{(3-\sqrt{3})h}{6} J) \,\, \exp(\tfrac{h}{2} \bar R) \,\, \exp(\tfrac{\sqrt{3}h}{3} J-\tfrac{(2-\sqrt{3})h^3}{24}C)\,\, \exp(\tfrac{h}{2} \bar R) \,\, \exp(\tfrac{(3-\sqrt{3})h}{6} J)
\end{align*}
The placement of the commutator $C$ in the central exponential as in (i) and (iii) results in one additional subproblem $f_c$ to be solved, whose system matrix, however, has no (skew-)symmetry properties in class~a), see (i). Due to the underlying port-Hamiltonian matrix structures, it might be convenient to incorporate the skew-symmetric commutator $C$ into the energy-conserving subproblem $f^{[2]}$. By setting up $f_c^{[2]}(x)=(J-ch^2C)x$, we then obtain commutator-based schemes that exclusively work on only two subproblems $f^{[1]}$, $f_c^{[2]}$ with inherent structures, see (ii) for class a).
\end{example}

\begin{example}[7-stage methods]\label{ex:4order7s} The increase of stages yields more degrees of freedom, and hence already for seven stages a large variety of fourth-order methods. They show a smaller coefficient $c$ in the force gradient term. Examples for both classes of dissipation-preserving symmetric 7-stage schemes are:\\[0.5ex]
\begin{tabular}{l l l l l l}
\quad (i) & $a_1=\tfrac{1}{8},$ & $a_2=\tfrac{3}{8},$ & $b_1=\tfrac{1}{3},$ & $b_2=\tfrac{1}{3},$   
& $c =e_{3_1}=\tfrac{1}{192}$\\ [0.4ex]
\quad (ii) & $a_2=\tfrac{3}{8},$ & $a_3=\tfrac{1}{4},$ & $b_1=\tfrac{1}{6},$ & $b_2=\tfrac{1}{3},$ & $c=e_{3_1}=\tfrac{1}{192}$
\end{tabular}\\[0.5ex]
The respective methods used for the forthcoming numerical simulations (cf.\ Section~\ref{sec:numerics}) contain the commutator in the energy-conserving subproblem.
\end{example}

\subsection{Sixth-order dissipation-preserving schemes}
A sixth-order scheme requires the satisfaction of all sixth-order conditions. Under the condition of positive step sizes, however, some of the terms $e_{5_\ell}$, $\ell=1,...,6$, cannot be canceled out at the same time  in \eqref{scheme:sym_order}, just like $e_{3_1}$ and $e_{3_2}$ before. Contradicting terms are, e.g., $e_{5_1}$ and $e_{5_4}$, as well as $e_{5_3}$ and $e_{5_6}$, as shown in \cite{chin2005structure}.

We proceed from the derived symmetric commutator-based fourth-order schemes \eqref{scheme:sym-force}, where the forced gradient term  $(-ch^3C)$ is incorporated in both subproblems, i.e.,
\begin{align*}
 \Phi_h &= \prod_{j=1}^m  \exp(a_jh L_c^{[1]})\,\exp(b_jh {L}_c^{[2]}), \\ 
L_c^{[1]} &=\bar R - \lambda^{[1]} c h^2 C, \qquad L_c^{[2]}=J - \lambda^{[2]} c h^2 C, \qquad \lambda^{[1]}+\lambda^{[2]}=1
\end{align*}
(cf.\ Example~\ref{ex:4order}, (ii)). The step size tuples $a, b\in \mathbb{R}_{\geq0}^m$ obey either case a) or case b) of \eqref{scheme:sym}.
According to \eqref{scheme:sym_order} the schemes can be expressed as exponential in terms of the original system matrices $L^{[1]}=\bar R$, $L^{[2]}=J$ and the respective skew-symmetric and symmetric commutators $C$ and $D$ (cf.\ Lemma~\ref{lem:com}). Thereby, the forced gradient yields modifications in higher orders,
\begin{align*}
       \Phi_h&=\exp(\,h\, (e_{1_1}{L}^{[1]}+e_{1_2} {L}^{[2]}) +h^3\,(\, (e_{3_1}-(\lambda^{[1]} e_{1_1}+\lambda^{[2]}e_{1_2})c)\,{C} + e_{3_2}{D}\,)\\ 
       & \quad \quad \, +
       h^5\, (\, (e_{5_1}-\lambda^{[2]}ce_{3_1}) [{L}^{[1]},[{L}^{[1]},{C}]]
       +(e_{5_2}+2\lambda^{[1]}c e_{3_1}+ \lambda^{[2]}c e_{3_2})    [{L}^{[1]},[L^{[1]},{D}]] \\
       & \quad \quad \qquad + \hspace*{1.65cm}
      e_{5_3}\,\,  [{L}^{[1]},[{L}^{[2]},{D}]]  
      + (e_{5_4}-\lambda^{[1]}c e_{3_1}-2 \lambda^{[2]}c e_{3_2})  [{L}^{[2]},[{L}^{[1]},{C}]] \\
       & \quad \quad \qquad +
       (e_{5_5} + \lambda^{[1]}c e_{3_2})  [{L}^{[2]},[{L}^{[2]},{C}]] + \hspace*{3.4cm} e_{5_6} \,\, [{L}^{[2]},[{L}^{[2]},{D}]]\,) + \mathcal{O}(h^7)\,).
 \end{align*}
Choosing the step size tuples $a$, $b$ such that $e_{1_1}(a,b)=e_{1_2}(a,b)=1$ and $e_{3_2}(a,b)=0$, we find the designed fourth order, as $c=e_{3_1}$. In $\mathcal{O}(h^5)$,  the force gradient affects then only the coefficients $e_{5_1}$, $e_{5_2}$ and $e_{5_4}$, since $e_{3_2}=0$. Despite the modification, however, the additional six conditions for sixth order cannot be satisfied simultaneously for positive step sizes. We observe contractions not only between the coefficients of the skew-symmetric commutators ($\tilde e_{5_1}, e_{5_3}, e_{5_5}$) and those of the symmetric commutators ($\tilde e_{5_2}, \tilde e_{5_4}, e_{5_6}$), but also between the coefficients of the symmetric commutators themselves. Non-vanishing skew-symmetric commutators can be straightforward incorporated into a scheme as forced gradients as done for fourth order. The treatment of the symmetric commutators, in contrast, is difficult. As they are in general non-definite, non-vanishing symmetric commutators initially either prevent the sixth order or violate the dissipation inequality when used as forced gradients.
However, we show that dissipation-preserving sixth-order schemes can be designed for special cases as well as by help of reformulations. 

Consider $\lambda^{[1]}=0$ and $\lambda^{[2]}=1$ in the following, then the skew-symmetric commutator $(-ch^3C)$ is exclusively contained in the energy-conserving subproblem.

\begin{scheme}[Schemes for special system matrices]\label{scheme:order6}
A dissipation-preserving sixth-order commutator-based splitting scheme for \eqref{model:autonom} is given by,
\begin{itemize}
\item[a)] if $[L^{[1]},{D}]=0$, a symmetric class-a scheme with $m\geq 5$, that is induced by a solution $a, b\in \mathbb{R}_{\geq0}^m$ of the non-linear system
\begin{align*}
e_{1_1}(a,b)=e_{1_2}(a,b)=1, \quad  \quad e_{3_2}(a,b)=0, \quad  \quad e_{5_4}(a,b)=e_{5_6}(a,b)=0
\end{align*}
with the skew-symmetric forced gradients $(-ch^3C)$, $c=e_{3_1}(a,b)$, as well as $(-h^5Z)$,
\begin{align*}
 Z=\tilde{e}_{5_1}(a,b)[{L}^{[1]},[{L}^{[1]},{C}]]+e_{5_3}(a,b)[{L}^{[1]},[{L}^{[2]},{D}]], \qquad \tilde{e}_{5_1}(a,b)=e_{5_1}(a,b)-c^2.
\end{align*}
\item[b)] if $[L^{[2]},{D}]=0$, a symmetric class-b scheme with $m\geq 5$, that is induced by a solution $a, b\in \mathbb{R}_{\geq0}^m$ of the non-linear system
\begin{align*}
e_{1_1}(a,b)=e_{1_2}(a,b)=1, \quad  \quad e_{3_2}(a,b)=0, \quad  \quad e_{5_2}(a,b)=e_{5_4}(a,b)=0
\end{align*}
with the skew-symmetric forced gradients $(-ch^3C)$, $c=e_{3_1}(a,b)$, as well as $(-h^5Z)$,
\begin{align*}
 Z=\tilde{e}_{5_1}(a,b)[{L}^{[1]},[{L}^{[1]},{C}]]+e_{5_5}(a,b)[{L}^{[2]},[{L}^{[2]},{C}]],\qquad \tilde{e}_{5_1}(a,b)=e_{5_1}(a,b)-c^2.
\end{align*}
\end{itemize}
 Note that the expressions $e_{k_\ell}(a,b)$ for the order conditions (and hence the non-linear systems to be solved) differ in the two scheme classes.
 
 For $m=5$ the step sizes and respective coefficients of the resulting 9-stage schemes are listed in Table~\ref{tab:coeff}.
\end{scheme}

\begin{table}[tb]
\begin{tabular}{rl || rl}
\multicolumn{4}{l}{\underline{Class a)}}\\
$a_1$ & 0.0741652386084523 & $e_{3_1}$ & $\phantom{-}3.4513723374828328 \cdot 10^{-3}$\\
$a_2$ & 0.3312015955219320 & $\tilde{e}_{5_1}$ & $-2.8754632332335723 \cdot 10^{-5}$\\
$a_3$ & 0.1892663317392310 & ${e}_{5_3}$  & $-7.1694217000153600 \cdot 10^{-6}$\\
$b_1$ & 0.2357603332527950\\
$b_2$ & 0.2642396667472050\\
\hline
\multicolumn{2}{l}{coefficients of error terms in $\mathcal{O}(h^5)$ if $[L^{[1]},{D}]\neq 0$:} &
$e_{5_2}$ & $-4.2236874448321610 \cdot 10^{-5}$ \\
\multicolumn{2}{l}{} & $e_{5_5}$ & $-1.3885239090889685 \cdot 10^{-5}$ \\
\\
\multicolumn{4}{l}{\underline{Class b)}}\\
$b_1$ & 0.0951068381417148 &  $e_{3_1}$ & $\phantom{-}2.7595070959696467\cdot 10^{-3}$\\
$b_2$ & 0.2665630297195250 & $\tilde{e}_{5_1}$ & $-1.4777636968353257\cdot 10^{-5}$\\
$b_3$ & 0.2766602642775210 &$e_{5_5}$ & $\phantom{-}1.3097476006184691\cdot 10^{-5}$\\
$a_2$ & 0.2216735783842150 & &\\
$a_3$ & 0.2783264216157850 & &\\
\hline
\multicolumn{2}{l}{coefficients of error terms in $\mathcal{O}(h^5)$ if $[L^{[2]},{D}]\neq 0$:} &
$e_{5_3}$ & $\phantom{-}1.2345801613546542\cdot 10^{-5}$ \\
\multicolumn{2}{l}{} & $e_{5_6}$ & $\phantom{-}1.5869088947135458 \cdot 10^{-5}$ \\[0.7ex]
\end{tabular}
\caption{9-stage class-a and class-b schemes: step sizes and coefficients of commutators, cf.\ Scheme~\ref{scheme:order6}. Symbolic computation of expressions $e_{k_\ell}(a,b)$ in Mathematica, solving of non-linear systems in Python with an accuracy of $\mathcal{O}(10^{-16})$. \label{tab:coeff}}
\end{table}

\begin{example}
A model example for a linear pH-system that satisfies $[L^{[1]},{D}]=0$ is the linear damped harmonic oscillator described by the differential equation 
$ m \ddot{q}=-d\dot{q}-kq$ for the state $q$ with parameters $m, d, k\geq 0$.
\end{example}

Beyond the special cases, we can set up sixth-order class-a schemes that are always dissipation-preserving. Proceeding from Scheme~\ref{scheme:order6}a), if $[L^{[1]},{D}]\neq 0$, the symmetric term $e_{5_2} [L^{[1]},[L^{[1]},{D}]]$ and the skew-symmetric term $e_{5_5} [L^{[2]},[L^{[2]},{C}]]$ (due to the Jacobi identity) appear in $\mathcal{O}(h^5)$. We shift the skew-symmetric term in the commutator $Z$ (force gradient) and rewrite  the symmetric term.
Using the Baker-Campbell-Hausdorff formula, we find
\begin{equation*}
    \exp(ahL^{[1]} - eh^5[L^{[1]},[L^{[1]},{D}]]) = \exp(\tfrac{ah}{2}L^{[1]} + \tfrac{2eh^4}{a}[L^{[1]},{D}])\,\, \exp(\tfrac{ah}{2}L^{[1]} - \tfrac{2eh^4}{a}[L^{[1]},{D}]) + \mathcal{O}(h^7)
\end{equation*}
as well as
\begin{equation*}
    \exp(ahL^{[1]} - eh^5[L^{[1]},[L^{[1]},{D}]]) = \exp(\tfrac{eh^4}{a}[L^{[1]},{D}]) \,\, \exp(ahL^{[1]}) \,\, \exp(-\tfrac{eh^4}{a}[L^{[1]},{D}]) + \mathcal{O}(h^9).
\end{equation*}
As $L^{[1]}$ and $D$ are symmetric, the commutator $[L^{[1]},{D}]$ is skew-symmetric. Thus, we can incorporate one of the above approximations (exponential product) as additional force gradient into the splitting scheme  without violating the dissipation inequality and obtain sixth order (cf.\ Scheme~\ref{scheme:order6-gen}).

\begin{scheme}[Class-a schemes]\label{scheme:order6-gen}
Given the step size tuples $a$, $b$ and coefficients $e_{k_l}(a,b)$ of Scheme~\ref{scheme:order6}a) for $m=5$, cf.\ Table~\ref{tab:coeff}, dissipation-preserving sixth-order symmetric class-a methods for \eqref{model:autonom} are 
\begin{align*}
 \Phi_h &= \prod_{j=1}^2  \exp(a_jh L^{[1]})\,\exp(b_jh {L}_{cz}^{[2]}) \,\,\, S \,\,\,  \prod_{j=1}^2 \exp(b_{3-j}h {L}_{cz}^{[2]}) \exp(a_{3-j}h L^{[1]}),\\ 
L^{[1]} &=\bar R , \qquad L_{cz}^{[2]}=J -  c h^2 C - h^4 Z, \quad c=e_{3_1}\\
& \hspace*{1.9cm} Z=\tilde{e}_{5_1}[{L}^{[1]},[{L}^{[1]},{C}]]
+e_{5_3}[{L}^{[1]},[{L}^{[2]},{D}]]
+e_{5_5}[{L}^{[2]},[{L}^{[2]},{C}]] \\
\intertext{where the term $S$ with $e=e_{5_2}$ can be one of the following}
(i)\phantom{i} \qquad S&=\exp(\tfrac{a_3h}{2}L^{[1]} + \tfrac{2eh^4}{a_3}[L^{[1]},{D}]) \,\,\, \exp(\tfrac{a_3h}{2}L^{[1]} - \tfrac{2eh^4}{a_3}[L^{[1]},{D}])\\
(ii) \qquad S&=\exp(\tfrac{eh^4}{a_3}[L^{[1]},{D}]) \,\,\, \exp(a_3hL^{[1]}) \,\,\, \exp(-\tfrac{eh^4}{a_3}[L^{[1]},{D}]). 
\end{align*}
\end{scheme}
\begin{remark}
Considering Scheme~\ref{scheme:order6-gen}, case (i) implies a 10-stage method where the system matrices of the two subproblems to be solved for $S$ have no (skew-)symmetry properties. Case~(ii) results in a 11-stage method. Hence, the solve of one further system is required, but here all occurring subproblems have inherent structures in consistency to the initial energy-associated decomposition. This might be advantageous in view of the numerical integration. 
\end{remark}

\section{Non-autonomous PH-Systems}\label{sec:nonaut}
In this section we present and explore splitting strategies for linear non-autonomous pH-systems of the form
\begin{align}\label{model:nonautonomous}
	\dot{x}(t) &= (J-R)\,x(t) + Bu(t), \qquad  
       && x(t_0) = x_0,\\
       y(t)&= B^Tx(t)\nonumber
\end{align}
with $u:[t_0,T]\rightarrow \mathbb{R}^p$ sufficiently smooth,  $J=-J^T$ and $R\geq 0$.
Special attention is paid to higher order accuracy and to the preservation of dissipation, i.e., $\tfrac{\mathrm{d}}{\mathrm{d}t}H(x(t)\leq y(t)^Tu(t)$ with $H(x)=\tfrac{1}{2}x^Tx$.

The intuitive extension of the energy-associated decomposition \eqref{model:autonom} to the non-autonomous setting by incorporating the time-dependent input into the dissipative subproblem allows the preservation of dissipation in case of Strang splitting ($p=2$), cf.\ \cite{frommer2023operator}, but not in higher order. The occurring commutators are differential operators (Lie derivatives) that lose the port-Hamiltonian structural properties.
We address this problem by considering a port-based decomposition. The splitting strategy has similarities with the "frozen time" technique discussed in \cite{blanes2006splitting,blanes2010splitting}. In addition, we extend the results from the energy-associated splitting of the autonomous case to the non-autonomous case by adapting the formula for the variation of constants. A similar procedure was investigated for methods of first and second order in \cite{faou2015analysis,ostermann2012error}.
 
\subsection{Port-based Decomposition}
\label{subsec: PBD}
We investigate a splitting based on a port-based decomposition for linear non-autonomous systems with general constant system matrix $A\in \mathbb{R}^{n\times n}$ and non-homogeneity $b:[t_0,T]\rightarrow \mathbb{R}^n$ sufficiently smooth. The decomposition is applied to the transformed autonomous version with $z=(x^T,t)^T$
\begin{align}\label{model:nonaut}
&\dot{z} = \begin{pmatrix} \dot{x} \\ \dot{t}  \end{pmatrix} = \begin{pmatrix} Ax + b(t) \\ 1  \end{pmatrix} =f^{[1]}(z)+f^{[2]}(z),  
&& {z}(0) = \begin{pmatrix} {x}_0 \\ {t}_0  \end{pmatrix},\\
&f^{[1]}(z)=\begin{pmatrix} b(t) \\ 0  \end{pmatrix}, \qquad  f^{[2]}(z)=\begin{pmatrix} Ax \\ 1  \end{pmatrix}.\nonumber
\end{align}
The Lie derivatives to the two subproblems are given by
\begin{align*}
     \mathcal{L}^{[1]} = \langle b(\cdot) , \nabla_x\rangle, \qquad \mathcal{L}^{[2]} = \langle A \,\cdot \,, \nabla_x\rangle + {\partial_ t}
 \end{align*}
 with the Euclidean scalar product $\langle \, \cdot \,, \, \cdot\,\rangle$.
In the context of port-Hamiltonian systems, we have  $b(t)=Bu(t)$ and $A=J-R$ with $J=-J^T$, $R\geq 0$, then the subproblem $f^{[1]}$ contains the impact of ports and  $f^{[2]}$ the inner dynamics.

\begin{lemma} \label{lem:C0}
The commutator $\mathcal{C}=[\mathcal{L}^{[1]},[\mathcal{L}^{[1]},\mathcal{L}^{[2]}]]$ vanishes for all functions $b$ and quadratic matrices $A$, i.e.,  $\mathcal{C}= 0$. Moreover, all commutators in higher order, except of $\mathcal{D}$, $[\mathcal{L}^{[2]},\mathcal{D}]$ and $ [\mathcal{L}^{[2]},[\mathcal{L}^{[2]},[...,[\mathcal{L}^{[2]},\mathcal{D}]]]]$, also vanish.
\end{lemma}

\begin{proof}
The Lie bracket is a linear differential operator. For any differentiable function $F:\mathbb{R}^{n+1}\rightarrow \mathbb{R}^m$ and $z=(x^T,t)^T\in \mathbb{R}^{n+1}$, we have
\begin{align*}
[\mathcal{L}^{[1]},\mathcal{L}^{[2]}]F(z) &=( \mathcal{L}^{[1]}\mathcal{L}^{[2]}-\mathcal{L}^{[2]}\mathcal{L}^{[1]})F(z)\\
&=\langle \ell(t), \nabla_x\rangle F(z), &&  \ell=-(\partial_t b-Ab).
\end{align*}
Hence, the commutator $[\mathcal{L}^{[1]},\mathcal{L}^{[2]}]=-[\mathcal{L}^{[2]},\mathcal{L}^{[1]}]$  only contains a $t$-dependent function and a gradient with respect to $x$, just as $\mathcal{L}^{[1]}$. This structure is handed over to $\mathcal{D}$ and to all other higher order commutators that consist of $\mathcal{L}^{[2]}$ and $\mathcal{D}$, i.e.,
\begin{align*}
\mathcal{D}=[\mathcal{L}^{[2]},[\mathcal{L}^{[1]},\mathcal{L}^{[2]}]]&=\langle k_0, \nabla_x\rangle,  && k_0\,=\partial_t \ell - A \ell,\\
\underbrace{[\mathcal{L}^{[2]},[\mathcal{L}^{[2]},[...,[\mathcal{L}^{[2]},}_{m-\text{times}}\mathcal{D}]]]]
&=\langle k_m, \nabla_x \rangle, &&
 k_m=\partial_t k_{m-1} - A k_{m-1}, \quad m=1,2,....
\end{align*}
Due to this structure, we obtain
\begin{align*}
\mathcal{C}=[\mathcal{L}^{[1]},[\mathcal{L}^{[1]},[\mathcal{L}^{[2]}]]=\sum_{i,j=1}^n b_i(\cdot ) \ell_j(\cdot) \,\partial_{x_i,x_j}-\sum_{i,j=1}^n \ell_i(\cdot) b_j(\cdot)  \,\partial_{x_i,x_j}=0,
\end{align*}
analogously all Lie brackets with $\mathcal{L}^{[1]}$ and another $t$-dependent differential operator in $x$ vanish. Commutators that contain $\mathcal{C}$ are zero, as $\mathcal{C}=0$.
\end{proof}

Due to the properties of the commutators, only three order conditions have to be satisfied to establish a fourth-order symmetric splitting scheme \eqref{scheme:sym_order}, i.e., $e_{1_1}(a,b)=e_{1_2}(a,b)=1$ and $e_{3_2}(a,b)=0$ for the step size tuples $a$, $b$. Increasing the order by two goes with only one further order condition, e.g., $e_{5_6}(a,b)=0$ in case of $p=6$. Hence, the desired order $p$ and the minimal number of required degrees of freedom in the scheme are related by $m= \tfrac{p}{2}+1$ for $p$ even. There is no necessity for forced gradients.

\begin{scheme}[Higher-order schemes]\label{scheme:port}\quad\\
$p=4$: The two symmetric 5-stage methods ($m=3$) for \eqref{model:nonaut} are given by, cf.\ Scheme~\ref{scheme:5stage},  \\[0.7ex]
\begin{tabular}{l l l l }
\,\, \qquad a) & $a_1=\tfrac{1}{6},$ & $a_2=\tfrac{2}{3},$ & $b_1=\tfrac{1}{2}$ \\[0.5ex]
\,\, \qquad b) & $a_2=\tfrac{1}{2},$ & $b_1=\tfrac{3-\sqrt{3}}{6},$ & $b_2=\tfrac{\sqrt{3}}{3}$ 
\end{tabular}\\[0.7ex]
$p=6$: The two symmetric 7-stage methods ($m=4$) for \eqref{model:nonaut} are given by \\[0.7ex]
\begin{tabular}{l l l l l}
\,\, \qquad a) & $a_1=\tfrac{1}{12},$ & $a_2=\tfrac{5}{12},$ & $b_1= \frac{5 - \sqrt{5}}{10},$ & $b_2 = \tfrac{\sqrt{5}}{5}$\\[0.5ex]
\,\, \qquad b) &  $a_2= \tfrac{5}{18},$ & $a_3 = \tfrac{4}{9},$ &
 $b_1 = \tfrac{5 - \sqrt{15}}{10},$ & $b_2= \tfrac{\sqrt{15}}{10}$ 
\end{tabular}\\[0.7ex]
Class-a schemes start with a solve of subproblem $f^{[1]}$ (hence $b_m=0$), class-b schemes with a solve of subproblem $f^{[2]}$ (hence $a_1=0$).
\end{scheme}

\begin{lemma} Let a linear non-autonomous pH-system of the form \eqref{model:nonautonomous} be given. Let $H$ be the Hamiltonian, $H(x)=\tfrac{1}{2}x^Tx$. Using a port-based decomposition, the splitting schemes with step size tuples $a$, $b\in \mathbb{R}^m_{\geq 0}$ given in Scheme~\ref{scheme:port} satisfy the dissipation inequality in an approximative quadrature-type version, i.e.,
\begin{align*}
 H(\Phi_{t_0+h,t_0}(x_0))-H(x_0) \leq d^{\mathrm{PBS}}_{t_0,h}:=\sum_{j=1}^m \int_{t_{j-1}^{[1]}}^{t_{j}^{[1]}} (y_j^{[1]}(t))^T \, \mathrm{d}t \,\, u(t_{j-1}^{[2]}).
 \end{align*}
Here, $y_j^{[1]}=B^T x_j^{[1]} $ denotes the output to the intermediate solutions of the first subproblem, regarding the underlying time grids $t_j^{[1]}=t_{j-1}^{[1]}+a_j h$ and  $t_j^{[2]}=t_{j-1}^{[2]}+b_j h$ with $ {t}_{0}^{[1]}={t}_{0}^{[2]}=t_0$.
\end{lemma}

\begin{proof}
A step of the splitting scheme is performed via the solves of $2m$ differential systems, i.e., $\Phi_{t_0+h,t_0}(x_0)=x_{m}^{[2]}(t_m^{[2]})$. Initialized with $x_{1}^{[1]}({t}_{0}^{[1]})=x_0$ and ${t}_{0}^{[1]}={t}_{0}^{[2]}=t_0$, the systems are given by, for $j=1,...,m$,
\begin{align*}
\dot x_{j}^{[1]}&=f^{[1]}_j , &&  x_{j}^{[1]}({t}_{j-1}^{[1]})=x_{j-1}^{[2]}({t}_{j-1}^{[2]}),\\
\dot x_{j}^{[2]}&=f^{[2]}_j, &&  x_{j}^{[2]}({t}_{j-1}^{[2]})=x_{j}^{[1]}(t_j^{[1]})
\end{align*} 
where $f_j^{[1]}=Bu(t_{j-1}^{[2]})$ and $f_j^{[2]}=(J-R)x_{j}^{[2]}$ as well as $t_j^{[1]}=t_{j-1}^{[1]}+a_j h$ and $t_j^{[2]}=t_{j-1}^{[2]}+b_j h$.

For the quadratic Hamiltonian we obtain
\begin{align*}
H(x_{m}^{[2]}(t_m^{[2]}))-H(x_0)=
\sum_{j=1}^m \bigg(H(x_{j}^{[2]}(t_j^{[2]}))-H(x_{j}^{[2]}(t_{j-1}^{[2]}))\bigg)+
\bigg(H(x_{j}^{[1]}(t_j^{[1]}))-H(x_{j}^{[1]}(t_{j-1}^{[1]}))\bigg),
\end{align*}
where the energy differences satisfy
\begin{align*}
H(x_{j}^{[1]}(t_j^{[1]}))-H(x_{j}^{[1]}(t_{j-1}^{[1]}))&=\int_{t_{j-1}^{[1]}}^{t_{j}^{[1]}} \tfrac{\mathrm{d}}{\mathrm{d}t}H(x_j^{[1]}(t)) \, \mathrm{d}t=\int_{t_{j-1}^{[1]}}^{t_{j}^{[1]}} (x_j^{[1]}(t))^T \, Bu(t_{j-1}^{[2]})
\, \mathrm{d}t,\\
H(x_{j}^{[2]}(t_j^{[2]}))-H(x_{j}^{[2]}(t_{j-1}^{[2]}))&=\int_{t_{j-1}^{[2]}}^{t_{j}^{[2]}} \tfrac{\mathrm{d}}{\mathrm{d}t}H(x_j^{[2]}(t)) \, \mathrm{d}t=\int_{t_{j-1}^{[2]}}^{t_{j}^{[2]}} (x_j^{[2]}(t))^T \, (J -R) \, x_j^{[2]}(t) \, \mathrm{d}t \leq 0
\end{align*}
for all $j$, since $J$ is skew-symmetric and $R$ is symmetric positive semi-definite. The result follows with $y_j^{[1]}(t)=B^T x_j^{[1]} (t)$ for $t\in [t_{j-1}^{[1]},t_{j}^{[1]}]$.
\end{proof}

\begin{example}
For the two 5-stage methods of Scheme~\ref{scheme:port}, the estimator $d_\mathrm{PBS}\approx \int_{t_0}^{t_0+h} y(t)^T u(t) \, \mathrm{d}t$ particularly reads 
\begin{align*}
\text{a)\quad }d^{\mathrm{PBS}}_{t_0,h} &= \int_{t_0}^{t_0+{h}/{6}}\hspace*{-0.3cm}(y_1^{[1]}(t))^T  \mathrm{d}t \,\, u(t_0)+ \int_{t_0+{h}/{6}}^{t_0+{5h}/{6}} \hspace*{-0.3cm}(y_2^{[1]}(t))^T  \mathrm{d}t \,\, u(t_0+\tfrac{h}{2})+\int_{t_0+{5h}/{6}}^{t_0+h}\hspace*{-0.3cm}(y_3^{[1]}(t))^T  \mathrm{d}t \,\, u(t_0+h)\\
\text{b)\quad }d^{\mathrm{PBS}}_{t_0,h} &= \int_{t_0}^{t_0+{h}/{2}}\hspace*{-0.05cm}(y_2^{[1]}(t))^T  \mathrm{d}t \,\, u(t_0+\tfrac{(3-\sqrt{3})h}{6})+ \int_{t_0+{h}/{2}}^{t_0+h} \hspace*{-0.05cm}(y_3^{[1]}(t))^T  \mathrm{d}t \,\, u(t_0+\tfrac{(3+\sqrt{3})h}{6})\\
\end{align*}
\end{example}

\begin{remark}
Using a port-based decomposition for a pH-system, the solves of the second subproblem for the inner dynamics can be performed with a commutator-based scheme of Section~\ref{sec:aut}.
\end{remark}

\subsection{Combination of Energy-associated Splitting and Quadrature}
\label{subsec: ESQ}
The exact solution of a linear non-autonomous system with general constant system matrix $A\in \mathbb{R}^{n\times n}$ and non-homogeneity $b:[t_0,T]\rightarrow \mathbb{R}^p$ sufficiently smooth is given by
\begin{align}\label{eq:exact}
    x(t) = \exp((t-t_0)A)\,x_0 + \int_{t_0}^t \exp((t - s)A)\,b(s)\ \mathrm{d}s.
\end{align}
Approximating the integral term with a quadrature rule, we obtain
\begin{align*}
    \Psi_{t_0+h,t_0}x_0 &= \exp(hA)x_0 + h\sum_{i=0}^\nu w_i \exp((h-s_i)A)\, b(t_0+s_i), \quad s_i\in[0,h].
    \end{align*}
 with quadrature nodes $s_i$ and weights $w_i$.
 A Simpson rule particularly yields a fourth-order approximation (here, $s_i=i\tfrac{h}{2}$, $i=0,1,2$, and $w_0=w_2=\tfrac{1}{6}$, $w_1=\tfrac{2}{3}$). 
 
 \begin{lemma}\label{scheme:nonau}
 Let a linear non-autonomous pH-system of the form \eqref{model:nonautonomous} be given. Let $H$ be the Hamiltonian, $H(x)=\tfrac{1}{2}x^Tx$. 
 Let $\Phi$ be a fourth-order commutator-based splitting scheme on top of an energy-associated decomposition of the respective homogeneous pH-system, where the skew-symmetric commutator is incorporated into the energy-conserving subproblem (cf.\ Scheme~\ref{scheme:5stage}),
 then 
\begin{align}\label{scheme:nonaut2}  \Psi^\Phi_{t_0+h,t_0}x_0 = \Phi_h (x_0+\tfrac{h}{6}Bu(t_0))+\tfrac{2h}{3}\Phi_{\frac{h}{2}} (Bu(t_0+\tfrac{h}{2})) +\tfrac{h}{6}Bu(t_0+h)
\end{align}
is a fourth-order scheme for the original (inhomogeneous) system.
 Moreover, the scheme satisfies the dissipation inequality in an approximative quadrature-type version, i.e.,
 \begin{align*}
 H(\Psi^\Phi_{t_0+h,t_0}x_0) - H(x_0) &\leq d^{\mathrm{ESQ}}_{t_0,h}:=\tfrac{1}{2} (\|x_0\|_2+\|\Psi^\Phi_{t_0+h,t_0}x_0\|_2) \,\|B\|_2 \, Q(\|u\|_2),
 \end{align*}
 with $Q(\|u\|_2)$ approximation of the integral $\int_{t_0}^{t_0+h} \|u(s)\|_2 \,\mathrm{d}s$ via Simpson quadrature rule.
 \end{lemma}
 
 \begin{proof}
 The exact solution is a superposition of the solution of the homogeneous system and of a particular solution,  cf.\ \eqref{eq:exact} with $A=J-R$ and $b(t)=Bu(t)$. Combining a Simpson rule for the particular solution (integral) with an energy-associated splitting for the homogeneous system yields the scheme in \eqref{scheme:nonaut2}, which is of fourth order due to the underlying fourth-order approximations.
 
Denoting
\begin{align*}
x_h=\Psi^\Phi_{t_0+h,t_0}x_0=\Phi_h x_0 +v, \qquad v=h\, (\tfrac{1}{6} \Phi_h Bu(t_0) + \tfrac{2}{3}\Phi_{\frac{h}{2}} Bu(t_0+\tfrac{h}{2}) +\tfrac{1}{6}Bu(t_0+h)),
\end{align*}
we obtain
\begin{align*} 
H(x_{h}) - H(x_0) &= \tfrac{1}{2}(\|x_h\|_2^2- \|x_0\|_2^2)\\
&=\tfrac{1}{2}( \|\Phi_h x_0\|_2^2- \|x_0\|_2^2) +\tfrac{1}{2} (\langle \Phi_h x_0, v\rangle+ \langle x_h,v\rangle ) \\
&\leq \tfrac{1}{2} (\| \Phi_h x_0\|_2 \|v\|_2+ \|x_h\|_2 \|v\|_2)\\
 & \leq  \tfrac{1}{2} (\|x_0\|_2+\|x_h\|_2) \,\|B\|_2 \, \bar u, \quad \bar{u}=h\, (\tfrac{1}{6} \|u(t_0)\|_2 + \tfrac{2}{3}\|u(t_0+\tfrac{h}{2})\|_2 +\tfrac{1}{6}\|u(t_0+h)\|_2).
\end{align*} 
We use here the fact that $\|\Phi_h\|_2\leq 1$ for our scheme $\Phi_h=\Pi_{j=1}^m \exp(a_j h(-R))\exp(b_j h(J-cC))$ with $C$ skew-symmetric, since $\|\exp(M)\|_2=1$ for $M$ skew-symmetric and $\|\exp(M)\|_2\leq 1$ for $M$ negative semi-definite.
 \end{proof}
 
 \begin{remark}
The fourth order in Lemma~\ref{scheme:nonau} could be alternatively also achieved if in \eqref{scheme:nonaut2} the half step $\Phi_{\frac{h}{2}}$ is performed with a third order (non-symmetric) commutator-based splitting scheme due to the factor $\tfrac{2h}{3}$, e.g., with the 4-stage method $\tilde{\Phi}$, \cite{faou2015analysis},
$$\tilde{\Phi}_h = \exp(\tfrac{h}{4}(L^{[2]}-\tfrac{h^2}{48}C))\exp(\tfrac{2h}{3}L^{[1]})\exp(\tfrac{3h}{4}(L^{[2]}-\tfrac{h^2}{48}C))\exp(\tfrac{h}{3} L^{[1]}),$$ using $L^{[1]}=\bar R$ and $L^{[2]}=J$. This saves the solve of one subsystem in every time step.
\end{remark}

\section{Numerical Results}\label{sec:numerics}
In this section we investigate the performance of the presented schemes. As benchmark examples act the damped linear oscillator without and with driving force as well as a model for rigid body dynamics. The numerical results  confirm the analytical findings concerning convergence order and dissipation inequality. 

\begin{testcase}[Driven damped linear oscillator]\label{test1}
The equation of motion for the damped linear oscillator with position $q\in \mathcal{C}^2(I,\mathbb{R})$ is given by
\begin{align*}
    m\ddot{q} = -d\dot{q} - kq-f(t)
\end{align*}
with $m, d, k\geq 0$. The driving force is (i) $f(t)=0$ and (ii) $f(t)=f_0 \cos(\omega t)$, implying either an autonomous or a non-autonomous system. We deal with the port-Hamiltonian formulation in terms of the state $x=(k^{1/2} q, m^{1/2} \dot q)^T$ on $I=[0,T]$ with input $u=f$
\begin{align*}
 \dot x&=(J-R)x+Bu, \quad x(0)=x_0, \qquad y=B^T x,\\
 &J=
\left (\begin{array}{c c} 0 &(\tfrac{k}{m})^{1/2}\\ -(\tfrac{k}{m})^{1/2} & 0
\end{array}\right), \quad R=\left(\begin{array}{c c}0 & 0 \\  0&\tfrac{d}{m} \end{array}\right), \quad B=\left (\begin{array}{c} 0 \\-(\tfrac{1}{m})^{1/2} \end{array}\right).
\end{align*}
In particular, we use here $m=d=1$, $k = 1000$ as well as $f_0= 5$, $\omega=3$, the initial value is set to $x_0 = (0,1)^T$, see Figs.~\ref{fig:ref} (left) and \ref{fig:nonauto}.
\end{testcase}

\begin{testcase}[Rigid body dynamics]\label{test2} The rigid body dynamics in three dimensions can be described in terms of the angular momentum $p$ around the principal axes. The Hamiltonian states the kinetic energy, $H(p)=\tfrac{1}{2} \sum_{i=1}^3 \tfrac{p_i^2}{I_i} =\tfrac{1}{2} p^TQp$ with moment of inertia $I_i>0$. In a linearized variant, equipped with friction $r_i\geq 0$, we consider 
\begin{align*}
 \dot p&=(J_Q-R_Q)Q\,p, \quad J_Q=
\left (\begin{array}{c c c} 0 &-1 & 1\\ 1 & 0 & -1 \\ -1 & 1 &0
\end{array}\right), \,\, R_Q=\left(\begin{array}{c c c}r_1 & 0 & 0 \\  0& r_2 & 0 \\ 0 & 0 & r_3 \end{array}\right), 
\,\, Q=\left(\begin{array}{c c c}\tfrac{1}{I_1} & 0 & 0 \\  0& \tfrac{1}{I_2} & 0 \\ 0 & 0 & \tfrac{1}{I_3}\end{array}\right). 
\end{align*}
The transformed version for $x=Q^{1/2}p$ is given by $\dot x=(J-R)x$ with $M=Q^{1/2}M_Q Q^{1/2}$ for $M\in \{J,R\}$. We use here $(r_1,r_2,r_3)=(0,5,1000)$ and $(I_1,I_2,I_3)=(1/4900,1,25)$ as well as the initial value $x(0)=(1,0,0)^T$, cf.\ Fig.~\ref{fig:ref} (right).
\end{testcase}

\begin{figure}[b]
\includegraphics[width=0.495\textwidth]{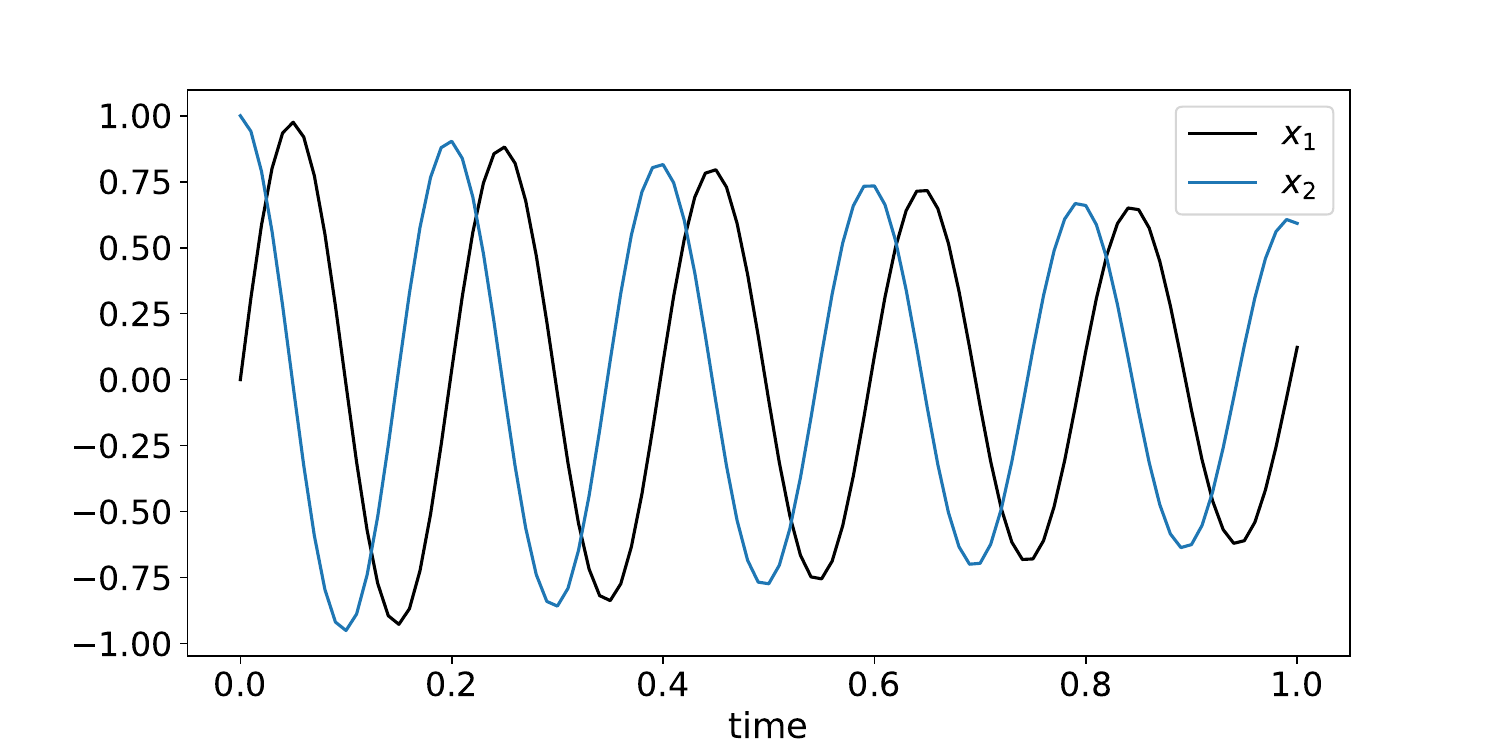}
\hfill
\includegraphics[width=0.495\textwidth]
{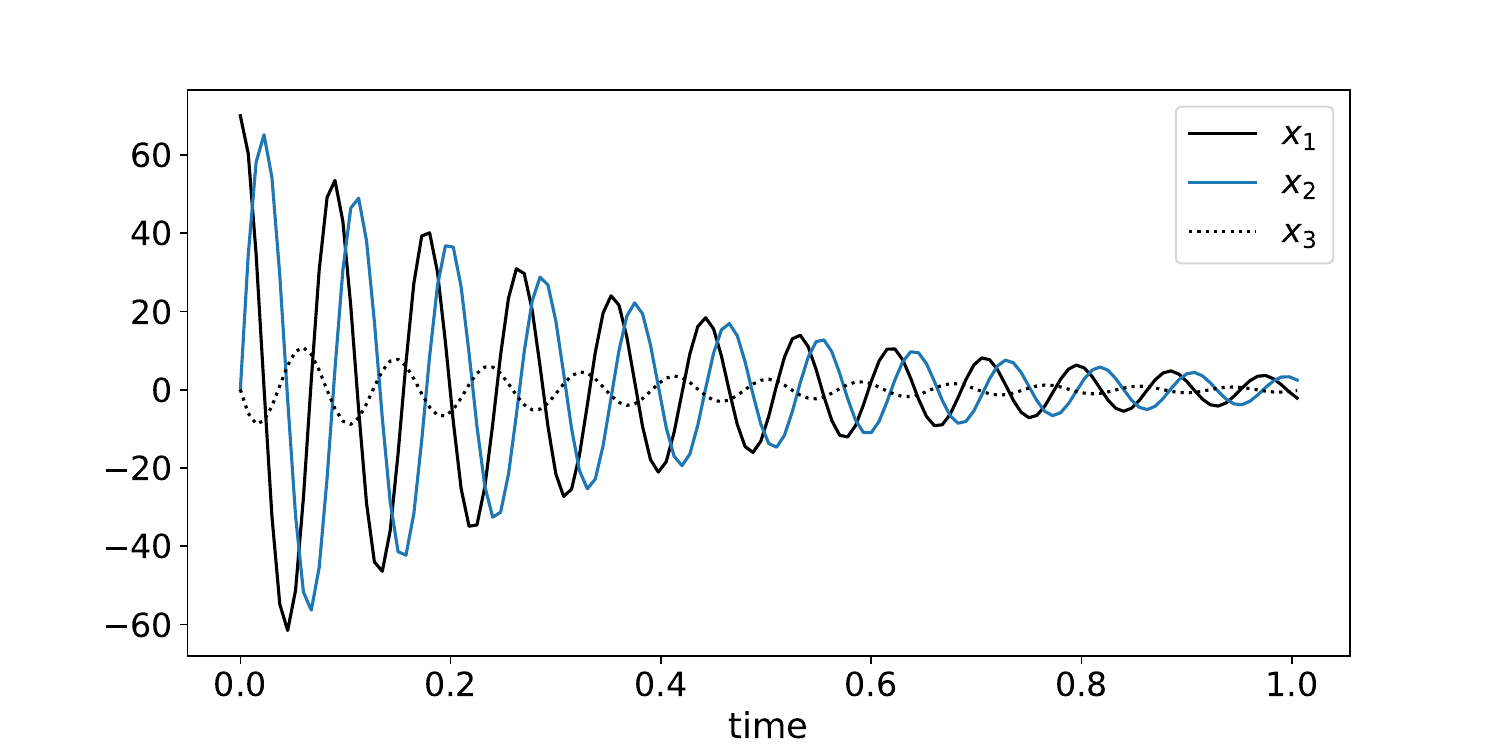} 
 \caption{References of autonomous systems. Left: benchmark test~\ref{test1}(i). Right: test~\ref{test2}.} 
 \label{fig:ref}
\end{figure}

\begin{remark}
The numerical computations are performed in Python. The exponential matrices are computed with the routine scipy.linalg.expm (Pade approximation of variable order based on the array data) and the integrals with scipy.integrate.quad (technique of Fortran library QUADPACK). The accuracy is set to $\mathcal{O}(10^{-16})$.
\end{remark}

\subsection{Performance of Energy-associated Splitting Schemes for Autonomous Systems}

\begin{figure}[tb]
\includegraphics[width=0.495\textwidth]{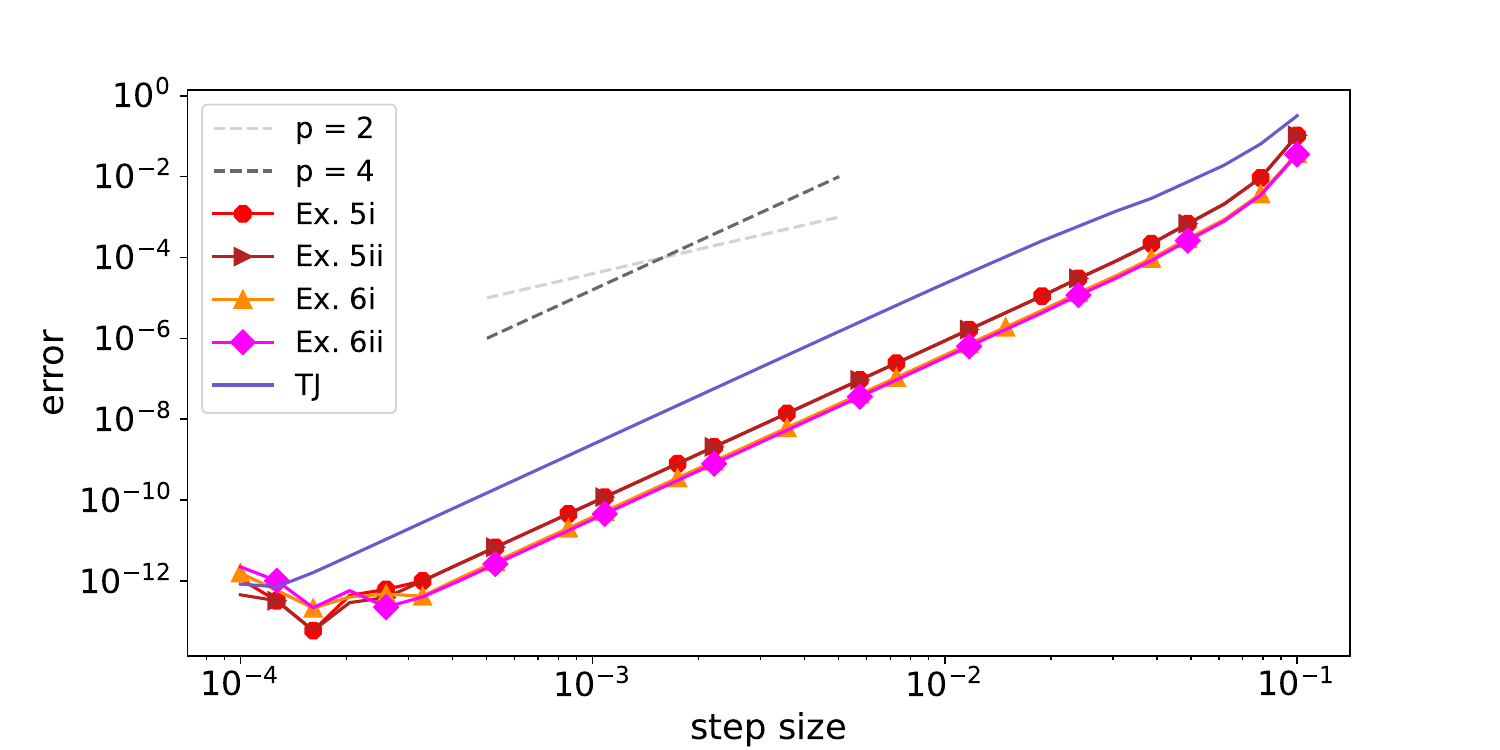}
\hfill
\includegraphics[width=0.495\textwidth]{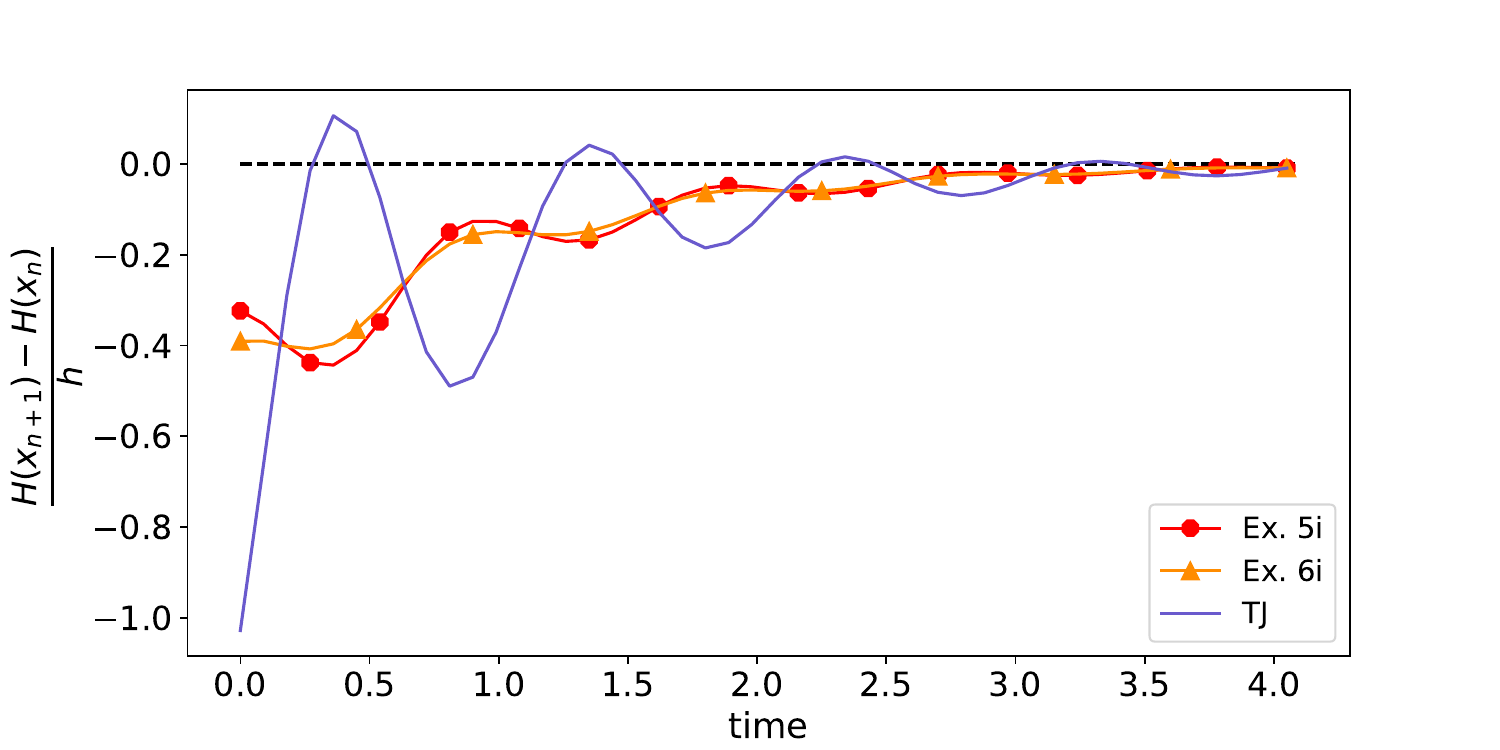}
 \caption{Energy-associated splitting. Convergence (left) and dissipation behavior $\tfrac{1}{h}({H(x_{n+1})-H(x_n)})$ for $h = 0.09$ (right) of some (fourth-order) commutator-based 5 and 7-stage methods from Examples~\ref{ex:4order} and \ref{ex:4order7s} and of the triple-jump method (TJ). Benchmark test~\ref{test1}(i).}
 \label{fig:auto_4order}
\end{figure}

\begin{figure}[tb]
\includegraphics[width=0.495\textwidth]{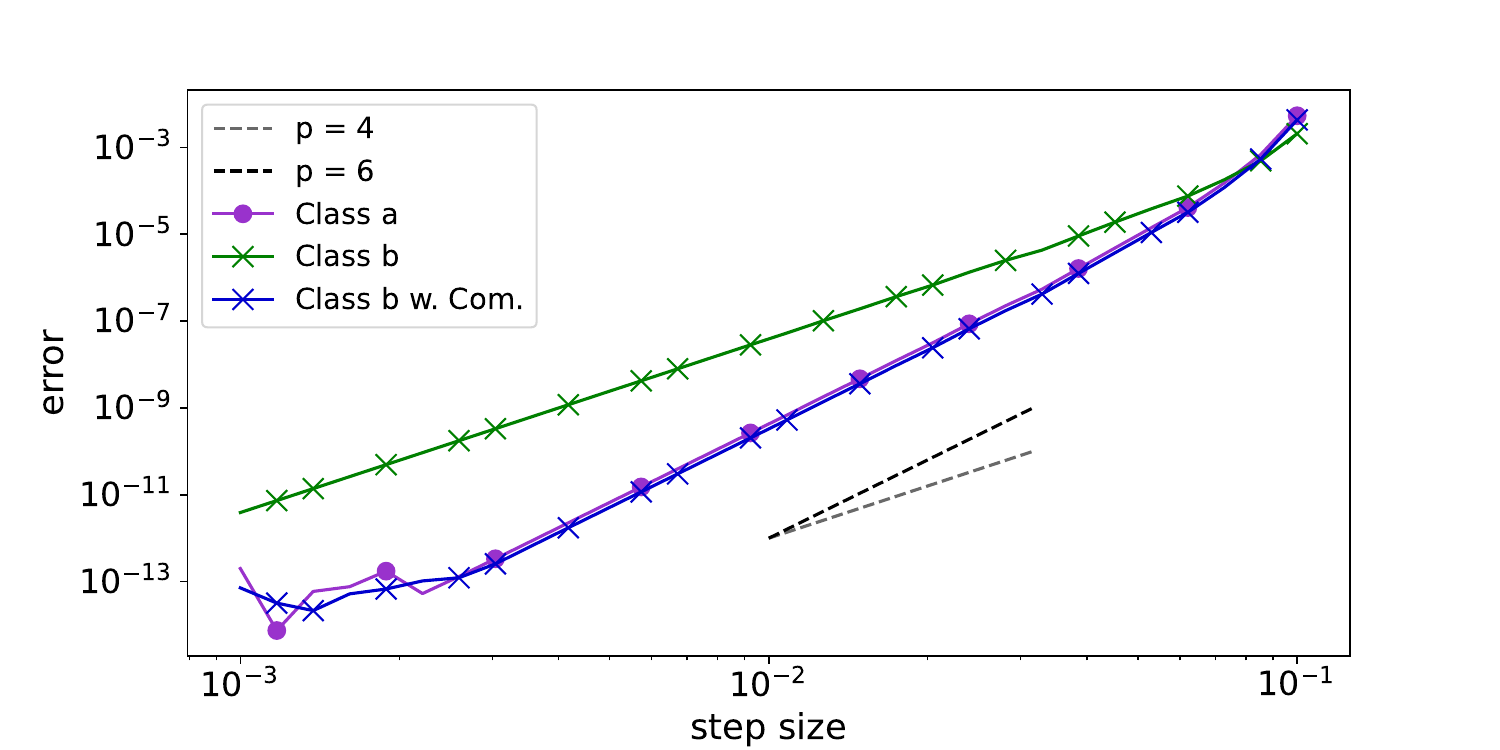}
 \hfill
\includegraphics[width=0.495\textwidth]{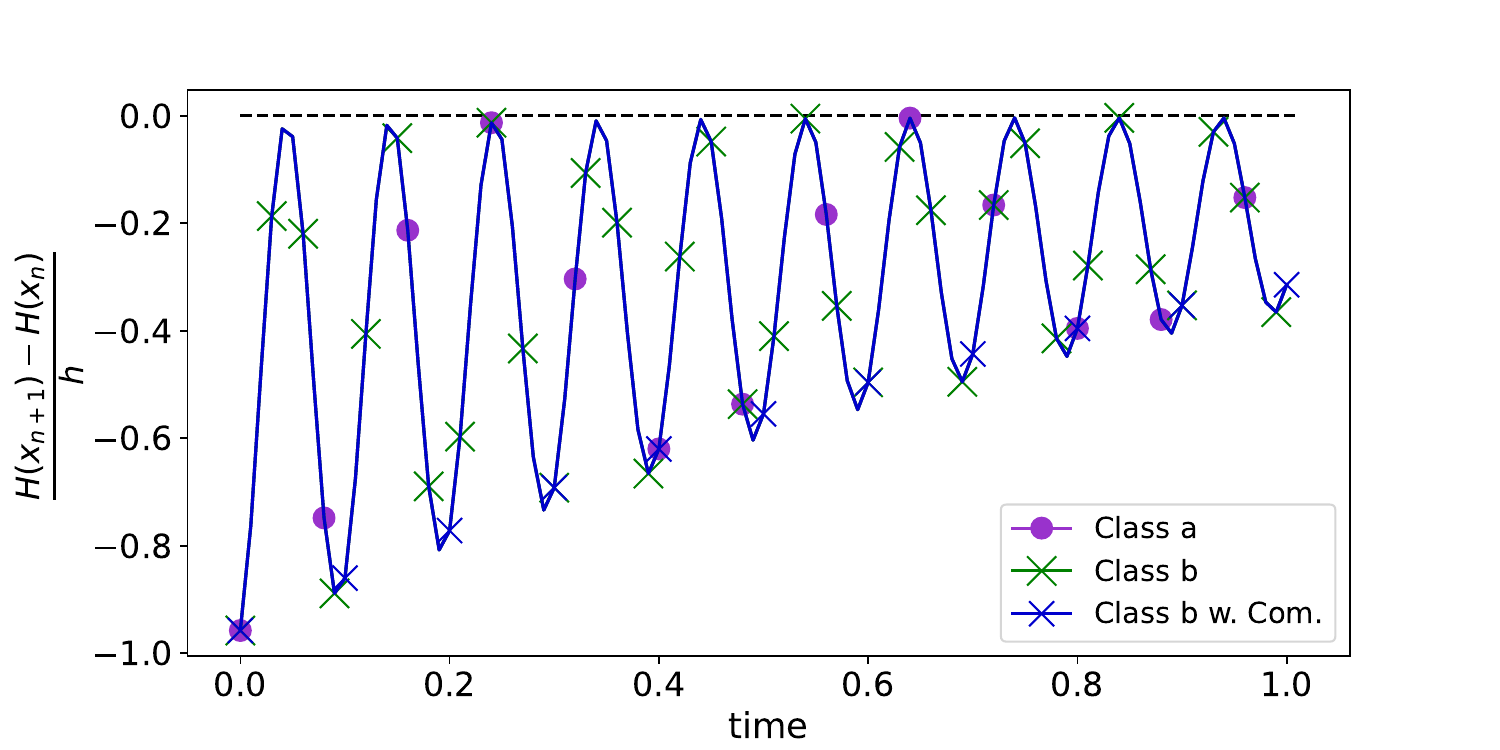}\\
\includegraphics[width=0.495\textwidth]{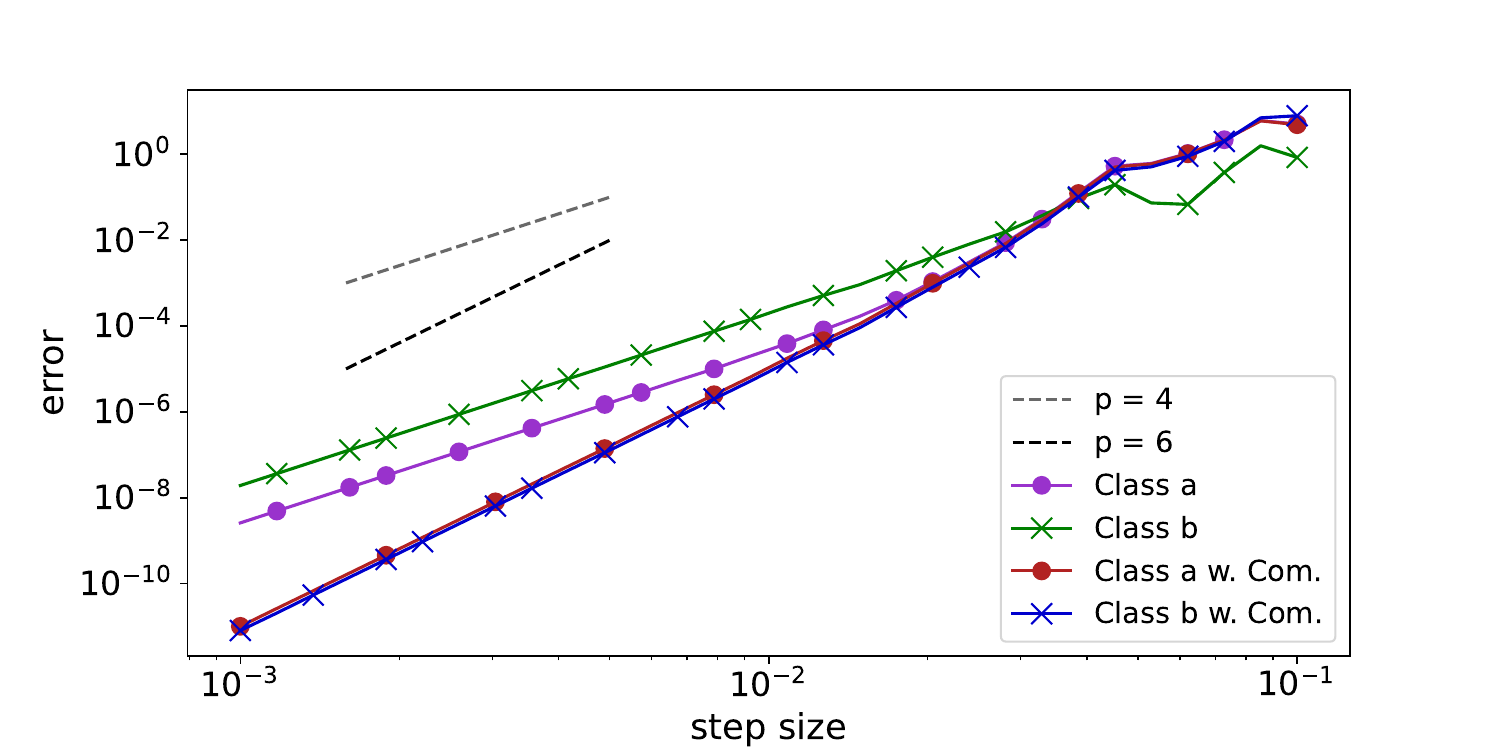}
 \hfill
\includegraphics[width=0.495\textwidth]{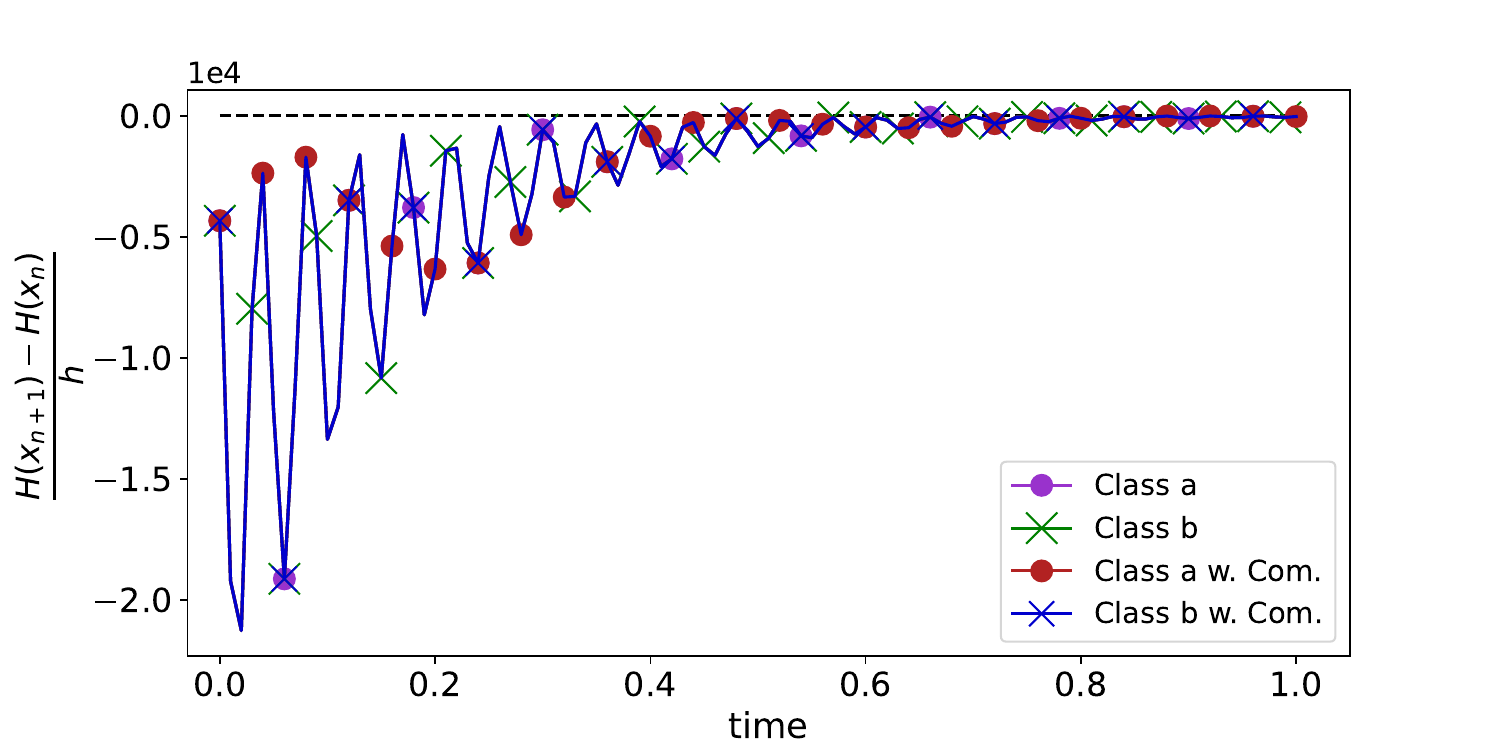}
 \caption{Energy-associated splitting. Convergence (left) and dissipation behavior for $h = 0.01$ (right) of 9-stage class-a and class-b methods from Scheme~\ref{scheme:order6} (original as well as modified with  5th order force gradients (w. Com.), cf.\ Table~\ref{tab:coeff}). Benchmark test~\ref{test1}(i) where $[L^{[1]},D]=0$ (top) and test~\ref{test2} (bottom).}
 \label{fig:auto_6order}
\end{figure}

The numerical performance of our derived fourth-order commutator-based schemes agrees with the theoretical results. The numerical convergence and dissipation behavior of some 5-stage and 7-stage methods taken from Examples \ref{ex:4order} and \ref{ex:4order7s} are visualized for the damped oscillator (test~\ref{test1}) in Figure~\ref{fig:auto_4order}. The results for the benchmark test~\ref{test2} are similar and hence omitted here.
The 5-stage methods show slightly higher error constants than the 7-stage methods, but the effort for one step $h$ is crucially less. Choice of scheme class and placement of the commutator (force gradient) might have impact on the computational effort in case of the flux approximations (numerical discretization of the subproblems), but not on the overall convergence and structure-preservation properties. In comparison, the classical triple-jump method with 7 stages and negative step sizes has a clearly higher error constant than our commutator-based schemes and violates the dissipation inequality in case of moderate $h$, i.e., $\tfrac{1}{h}H(x_{n+1}-H(x_n))\not \leq 0$ at certain time points $t_n$. However, note that in the limit $h\rightarrow 0$ the dissipation inequality is also satisfied here.

The damped oscillator in benchmark test~\ref{test1} is modeled by a special pH-system, where the system matrices imply a vanishing commutator, i.e., $[L^{[1]},D]=0$. In accordance to our theoretical investigations, the 9-stage class-a method of Scheme~\ref{scheme:order6}a), $m=5$, cf. Table~\ref{tab:coeff}, shows here numerically sixth-order convergence and the preservation of the dissipation inequality. The associated 9-stage class-b method acts also dissipation-preserving, but only fourth-order convergent. Adding the non-vanishing fifth-order commutators as forced gradient terms, the modified class-b method becomes sixth-order convergent, but violates the dissipation inequality with an error of $\mathcal{O}(h^5)$, cf.\ Fig.~\ref{fig:auto_6order} (top). 
In case of general pH-systems, the performance of the class-b method (with and without modification) is representative for both classes of Scheme~\ref{scheme:order6} as visualized  for benchmark test~\ref{test2} in Fig.~\ref{fig:auto_6order} (bottom).
The methods of Scheme~\ref{scheme:order6-gen}, in contrast, yield numerically sixth-order convergence and the preservation of the dissipation.

\subsection{Numerical Treatment of Non-autonomous Systems}
For the non-autonomous system of benchmark test \ref{test1}(ii), the numerical convergence and dissipation behavior of the port-based splitting (PBS) and the energy-split--quadrature strategy (ESQ) are illustrated in Fig~\ref{fig:nonauto}. The port-based splitting allows for symmetric methods of fourth, sixth and higher even orders. Figure~\ref{fig:nonauto} shows exemplarily the respective class-a methods of Scheme~\ref{scheme:port}. We observe the fourth and sixth order of convergence as well as the satisfaction of the dissipation inequality. The estimator $d^\mathrm{PBS}$ agrees with the (exact) supplied energy with an approximation error of $\mathcal{O}(h)$ in any order.
The energy-split--quadrature strategy is evaluated on top of an energy-associated decomposition, in particular the 5-stage class-a method from Example~\ref{ex:4order}(ii) is used here. In the numerical simulation we obtain the designed fourth order. Moreover, the approximated energy evolution is bounded by the exact supplied energy as desired. However, note that the derived energy estimator $d^\mathrm{ESQ}$ is not meaningful due to its general positivity, as seen in Fig.~\ref{fig:nonauto} (bottom, right).

\begin{figure}[tb]
\includegraphics[width=0.495\textwidth]
 {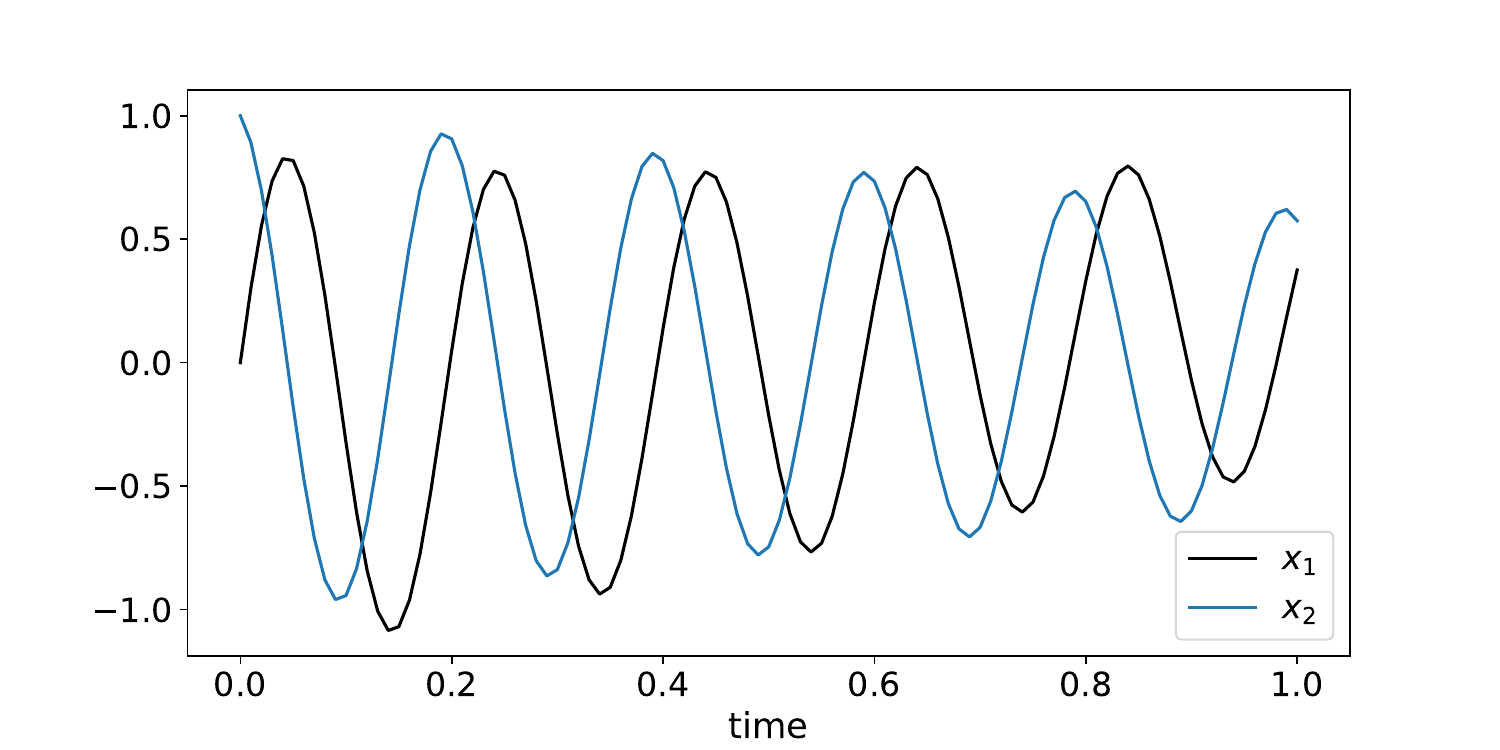}
 \hfill
\includegraphics[width=0.495\textwidth]
{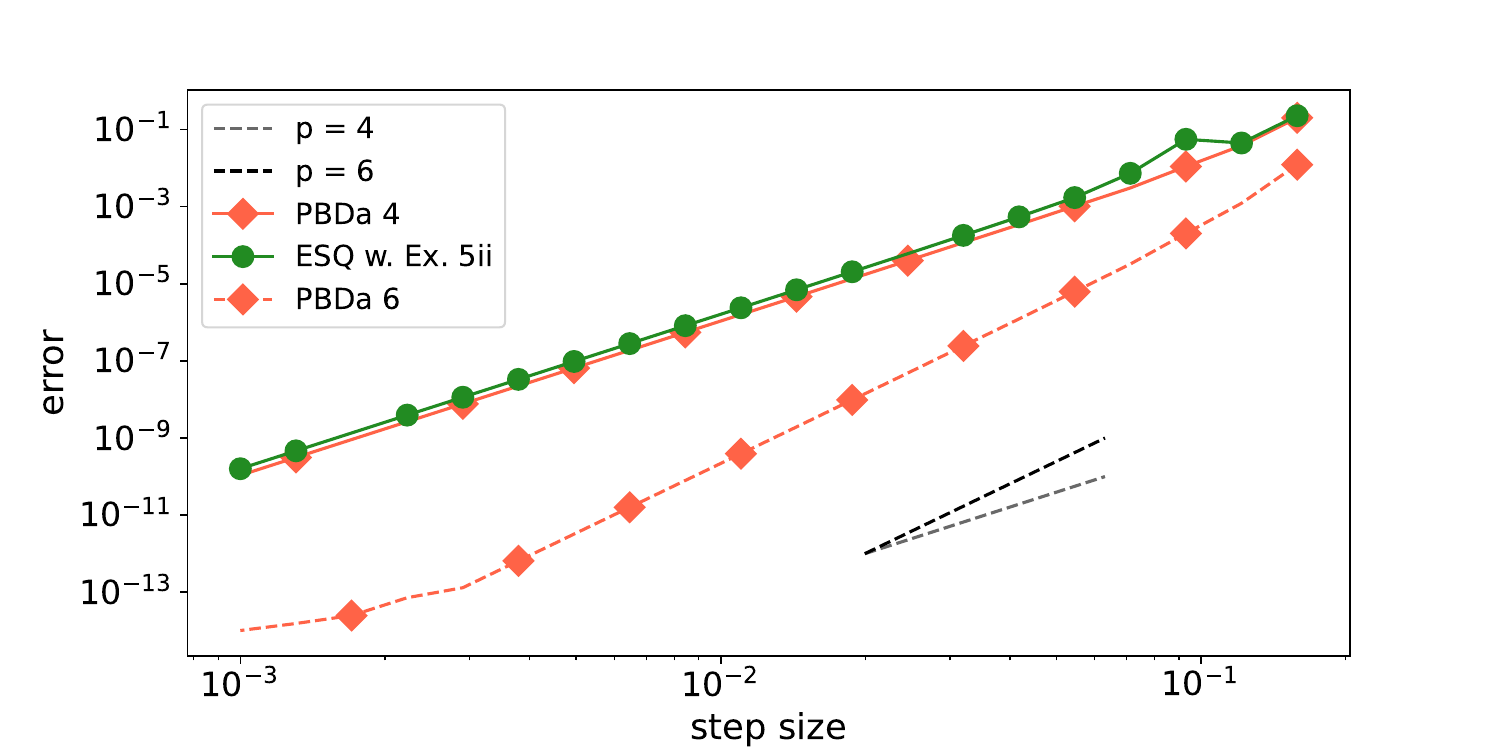}\\
\includegraphics[width=0.495\textwidth]{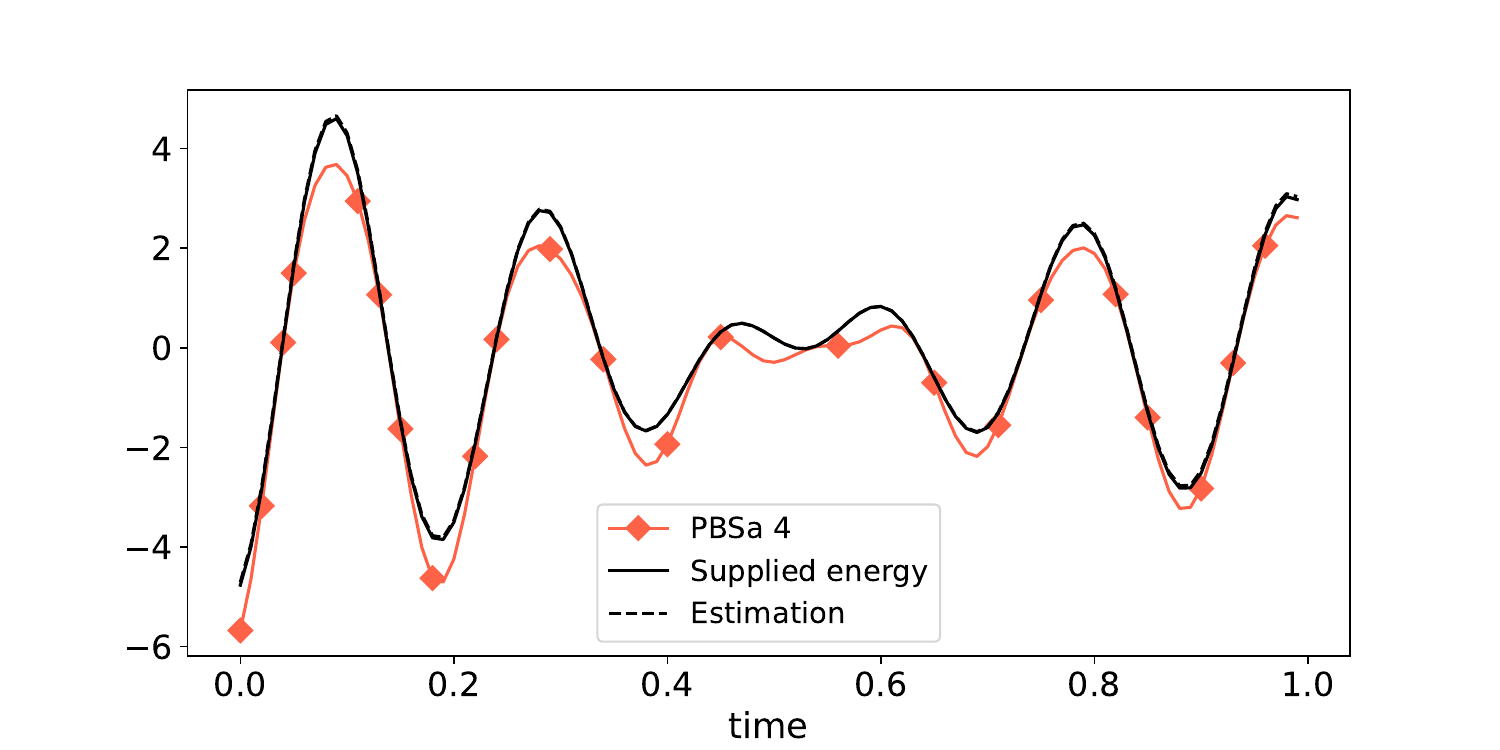}
 \hfill
\includegraphics[width=0.495\textwidth]{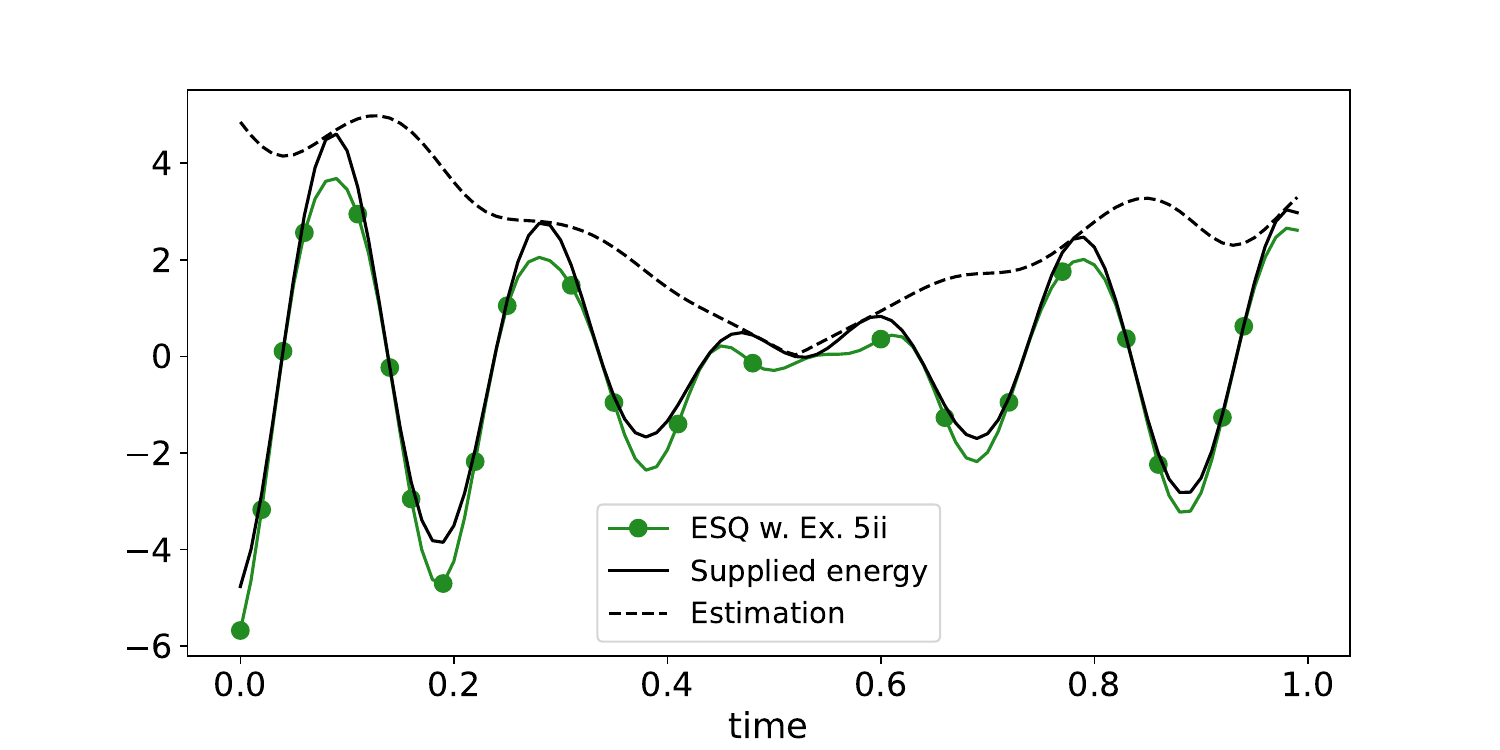}
\caption{Splitting for non-autonomous systems. Top: reference of benchmark test~\ref{test1}(ii) (left); convergence behavior of port-based splitting (PBS) and energy-split--quadrature strategy (ESQ) (right). Bottom: dissipation behavior $\tfrac{1}{h}({H(x_{n+1})-H(x_n)})$ for $h = 0.01$ with (exact) supplied energy $y^T u$ and respective estimator $\tfrac{d}{h}$ for PBS (left) and ESQ (right).}
\label{fig:nonauto}
\end{figure}

\section{Conclusion}
In this paper, we developed and analyzed splitting methods for linear port-Hamiltonian systems, distinguishing between autonomous and non-autonomous settings.

For autonomous systems, we introduced commutator-based methods grounded in an energy-as\-so\-ci\-ated decomposition. Specifically, the pH-system is decomposed into the skew-symmetric (structure) matrix $J$ and the symmetric positive semidefinite (dissipation) matrix $R$. This decomposition yields commutators that possess properties of skew-symmetry and indefinite symmetry. Utilizing these properties, we constructed splitting methods of up to fourth order with strictly positive step sizes. For special system matrices, we set up 9-stage sixth-order methods, and through the use of linear combinations, established generally applicable 10- and 11-stage methods of order six. The skew-symmetry inherent in the commutators ensured the preservation of the dissipation inequality, enabling the development of structure-preserving splitting methods that maintain the physical integrity of the pH-system.

In the non-autonomous case, we observed that commutators generally lose their structural properties across most decompositions. Consequently, we focused on a port-based decomposition, employing a \ldq frozen time\rdq\ approach that allows for the separate consideration of port effects and internal dynamics. The key advantage of this decomposition is that it simplifies the order conditions by ensuring that only a single commutator appears at any given order. This simplification facilitates the construction of higher-order methods, which can achieve arbitrary order without the inclusion of commutators within the splitting scheme. The dissipation inequality in this context provides an accurate approximation of the system’s exact energy, especially for sufficiently small step sizes. Additionally, we discussed the applicability of the autonomous methods in the non-autonomous framework.

Although the developed methods are computationally demanding due to their reliance on exponential mappings, practical implementations can be made more feasible through the use of appropriate flux approximations. Numerical integrators can reduce computational costs while preserving the accuracy and structural integrity of the solutions.
\nocite{ForcegradientPH}
\bibliographystyle{abbrv}
\bibliography{Commutatorbased}

\end{document}